\documentclass[conference]{IEEEtran}
\pdfoutput=1
\usepackage{amsmath}
\usepackage{amsfonts}
\usepackage{amsthm}

\usepackage{mathrsfs}
\usepackage{amssymb}

\usepackage{algorithm}
\usepackage{algpseudocode}

\usepackage{url}
\usepackage{comment}
\usepackage{float}
\newcommand{\eat}[1]{}
\theoremstyle{definition}
\usepackage{times}
  \renewenvironment{thebibliography}[1]{%
    \begin{oldthebibliography}{#1}%
      \setlength{\parskip}{0ex}%
      \setlength{\itemsep}{0ex}%
  }%
  {%
    \end{oldthebibliography}%
  }
  
  \newcommand{\squishlist}{
 \begin{list}{$\bullet$}
  { \setlength{\itemsep}{0pt}
     \setlength{\parsep}{3pt}
     \setlength{\topsep}{3pt}
     \setlength{\partopsep}{0pt}
     \setlength{\leftmargin}{1.0em}
     \setlength{\labelwidth}{1em}
     \setlength{\labelsep}{0.5em} } }

\newcommand{\squishlisttwo}{
 \begin{list}{$\bullet$}
  { \setlength{\itemsep}{0pt}
    \setlength{\parsep}{0pt}
    \setlength{\topsep}{0pt}
    \setlength{\partopsep}{0pt}
    \setlength{\leftmargin}{2em}
    \setlength{\labelwidth}{1.5em}
    \setlength{\labelsep}{0.5em} } }

\newcommand{\squishend}{
  \end{list}  }
  
\ifCLASSINFOpdf
   \usepackage[pdftex]{graphicx}
\else
   \usepackage[dvips]{graphicx}
\fi

\usepackage{algorithm}
\usepackage[tight,footnotesize]{subfigure}

\begin{document}
%
\title{Efficient Routing for Cost Effective Scale-out\\ Data Architectures}

\author{\IEEEauthorblockN{Ashwin Narayan}
\IEEEauthorblockA{Williams College\\ 
\footnotesize{ashwin.narayan@williams.edu}}
\and
\IEEEauthorblockN{Vuk Markovi\'c}
\IEEEauthorblockA{University of Novi Sad\\
\footnotesize{mvukmarko@gmail.com}}
\and
\IEEEauthorblockN{Natalia Postawa}
\IEEEauthorblockA{Adam Mickiewicz University in Pozna\'n\\
\footnotesize{np96924@st.amu.edu.pl}}
\and
\IEEEauthorblockN{Anna King}
\IEEEauthorblockA{University College Cork\\
\footnotesize{111396801@umail.ucc.ie}}
\and
\IEEEauthorblockN{Alejandro Morales}
\IEEEauthorblockA{University of California, Los Angeles\\
\footnotesize{ahmorales@math.ucla.edu}}
\and\and\and\and
\IEEEauthorblockN{K.~Ashwin Kumar}
\IEEEauthorblockA{Veritas Labs\\
\footnotesize{ashwin.kayyoor@veritas.com}}
\and\and\and\and
\IEEEauthorblockN{Petros Efstathopoulos}
\IEEEauthorblockA{Symantec Research Labs\\
\footnotesize{petros\_efstathopoulos@symantec.com}}
}

\newcommand{\topic}[1]{\vspace{0pt} \noindent {\underline{\bf #1}}}


%


\maketitle

\begin{abstract}
Efficient retrieval of information is of key importance when using Big Data systems. In large scale-out data architectures, data are distributed and replicated across several machines. Queries/tasks to such data architectures, are sent to a router which determines the machines containing the requested data. 
Ideally, to reduce the overall cost of analytics, the smallest set of machines required to satisfy the query should be returned by the router. 
Mathematically, this can be modeled as the \emph{set cover} problem, which is NP-hard\eat{ (computationally intractable)}, thus making the routing process a balance between optimality and performance. 
Even though an efficient greedy approximation algorithm for routing a single query exists, there is currently no better method for processing multiple queries than running the greedy set cover algorithm repeatedly for each query. 
This method is impractical for Big Data systems and the \eat{current baseline method of choice}state-of-the-art techniques route a query to all machines and choose as a cover the machines that respond fastest. 
In this paper, we propose an efficient technique to speedup the routing of a 
large number of real-time queries while minimizing the number of machines that each query touches ({\em query span})\eat{ in both real-time and non-real-time cases}. 
We demonstrate that \eat{precomputing, using variants of known clustering algorithms, can be used to cluster queries for more efficient routing}by analyzing the correlation between known queries and performing query clustering, we can reduce the set cover computation time, thereby significantly speeding up routing of unknown queries.\eat{Furthermore, for the real-time case, where we do not know the queries beforehand, we show that our algorithms can analyze queries in real-time and process them to improve the routing as queries arrive.} Experiments show that our incremental set cover-based routing is $2.5$ times faster and can return on average $50\%$ fewer machines per query when compared to repeated greedy set cover and baseline routing techniques.
\eat{Processing the queries as clusters reduces the set cover time cost}\eat{Our results can contribute to speeding up the routing of multiple queries as our algorithm is closer to optimal than baseline, and incurs lower run-time than the ``repeated greedy'' approach.}
\end{abstract}

\IEEEpeerreviewmaketitle

\section{Introduction}\label{Ch:Introduction}

Large-scale data management and analysis is rapidly gaining importance because of an exponential increase in the data volumes being generated in a wide range of application domains. The deluge of data (popularly called ``Big Data") creates many challenges in storage, processing, and querying of such data. These challenges are intensified further by the overwhelming variety in the types of applications and services looking to make use of Big Data. There is growing consensus that a single system cannot cater to the variety of workloads, and different solutions are being researched and developed for different application needs. For example, column-stores are optimized specifically for data warehousing applications, whereas row-stores are better suited for transactional workloads. There are also hybrid systems for applications that require support for both transactional workloads and data analytics. Other systems are being developed to store different types of data, such as document data stores for storing XML or JSON documents, and graph databases for graph-structured or RDF data.

One of the most popular approaches employed in order to handle the increasing volume of data is to use a cluster of commodity machines to parallelize the compute tasks (\emph{scale-out} approach). Scale-out is typically achieved by partitioning the data across multiple machines. Node failures present an important problem for scale-out architectures resulting in data unavailability. In order to tolerate machine failures and to improve data availability, data replication is typically employed. Although large-scale systems deployed over scale-out architectures enable us to efficiently address the challenges related to the volume of data, processing speed and and data variety, we note that these architectures are prone to resource inefficiencies.  Also, the issue of minimizing resource consumption in executing large-scale data analysis tasks is not a focus of many data systems that are developed to date. In fact, it is easy to see that many of the design decisions made, especially in scale-out architectures, can typically reduce overall execution times, but can lead to inefficient use of resources~\cite{Kumar:2014:SWD:2691523.2691545}\cite{Quamar:2013:SSW:2452376.2452427}\cite{ashwin_phd}.  As the field matures and the demands on computing infrastructure grow, many design decisions need to be re-visited with the goal of minimizing resource consumption. Furthermore, another impetus is provided by the increasing awareness that the energy needs of the computing infrastructure, typically proportional to the resource consumption, are growing rapidly and are responsible for a large fraction of the total cost of providing the computing services. 

To minimize the scale-out overhead, it is often useful to control the unnecessary spreading out of compute tasks across multiple machines. Recent works~\cite{Kumar:2014:SWD:2691523.2691545}\cite{Quamar:2013:SSW:2452376.2452427}\cite{ashwin_phd}\cite{Curino:2010:SWA:1920841.1920853}\cite{Kulkarni:2015:SSE:2766484.2738035} have demonstrated that minimizing the number of machines that a query or a compute task touches (\emph{query span}) can achieve multiple benefits, such as: \emph{minimization of communication overheads}, \emph{lessening total resource consumption}, \emph{reducing overall energy footprint} and \emph{minimization of distributed transactions costs}.
\eat{\begin{itemize}
\setlength\itemsep{0.2em}
\item \emph{Minimization of communication overheads}
\item \emph{Lessens total resource consumption}
\item \emph{Reduces overall energy footprint}
\item \emph{Minimizes the cost of distributed transactions}
\end{itemize}}

\subsection{The Problem}

In a scaled-out data model, when a query arrives to the query router, it is forwarded to a subset of machines that contain the data items required to satisfy the query. \eat{Naturally Symantec wants to optimize their ability to route this data and handle these millions of queries efficiently, thus avoiding overloading machines and conserving energy.} In such a data setup, a query is represented as the subset of data needed for its execution. As the data are distributed, this implies that queries need to be routed to multiple machines hosting the necessary data. To avoid unnecessary scale-out overheads, the size of the set of machines needed to cover the query should be minimal~\cite{Kumar:2014:SWD:2691523.2691545}\cite{Quamar:2013:SSW:2452376.2452427}\cite{ashwin_phd}. Determining such a minimal set is  mathematically stated as the \emph{set cover} problem, which is an NP-hard problem\eat{(see \cref{Section:backgoundsc})}. The most popular approximate solution of the set cover problem is a \textit{greedy algorithm}\eat{(greedy algorithms are a class of algorithms that look for ``locally'' optimal choices while trying to build an optimal solution)}. However, running this algorithm on each query can be very expensive or unfeasible when several million queries arrive all at once or in \textit{real-time} (one at a time) at machines with load constraints. Therefore, in order to speed up the routing of queries, we want to reuse previous {\em set cover} computations across queries without sacrificing optimality. In this work, we consider a generic model where a query can be either a database query, web query, map-reduce job or any other task that touches set of machines to access multiple data items. 


There is a large amount of literature available on single query {\em set cover} problems (discussed in Section~\ref{sec:related_work}). However, little work has been done on sharing {\em set cover} computation across multiple queries. As such, our main objective is to design and analyze different algorithms that efficiently perform {\em set cover} computations on multiple queries. We catalogue any assumptions used and provide guarantees relating to them where possible, with optimality discussions and proofs provided where necessary to support our work. We conducted an extensive literature review looking to understand algorithms, data structures and precomputing possibilities. This included algorithms such as  \textit{linear programming} \cite{vaz13}, data structures such as {\em hash tables}, and precomputing possibilities such as {\em clustering}. 
We built appropriate models for analysis that included a model with nested queries, and a model with two intersecting queries. These models afforded us a better understanding of the core of the problem, while also allowing us to test the effectiveness of our methods and tools. We developed an algorithm that will solve {\em set cover} for multiple queries more efficiently than repeating the greedy {\em set cover} algorithm for each query, and with better optimality than the current algorithm in use (see Section~\ref{sec:base}). We propose algorithm frameworks that solve the problem and experimentally show that they are both fast and acceptably optimal. \eat{For the non-real-time framework, we first cluster all the queries, and then process the 
clusters to return the covers for each of the queries. For the real-time framework, we}Our framework essentially analyzes the history of queries to cluster them\eat{a fraction of the queries (the ones which we know)} and use that information to process the new incoming queries in real-time. Our evaluation of the clustering and processing algorithms for both frameworks show that both sets are fast and have good optimality. 

The key contributions of our work are as follows:
\squishlist
\item To the best of our knowledge, our work is the first to enable sharing of {\em set cover} computations across the input sets (queries) in real-time and amortize the routing costs for queries while minimizing the average {\em query span}. 
\eat{\item To the best of our knowledge, our work is the first to make use of information about previous set cover computations of subsets (data accessed by a particular query in our case) of a particular dataset for improving the speed of new computations of set covers for the new subsets (data accessed by new queries)\eat{ of the same dataset}. We term this problem as the {\em incremental set cover} problem.}
\item We systematically divide the problem into three phases: {\em clustering the known queries}, {\em finding their covers}, and, with the information from the second phase, {\em covering the rest of the queries} as they arrive {\em in real time}. Using this approach, each of the three phases can then be improved separately, therefore making the problem easier to tackle.
\item We propose a novel entropy based real-time clustering algorithm to cluster the queries arriving in real-time to solve the problem at hand. Additionally, we introduce a new variant of greedy {\em set cover} algorithm that can cover a query $Q_i$ with respect to another correlated query $Q_j$.
\item Extensive experimentation on real-world and synthetic datasets shows that our incremental {\em set cover}-based routing is $2.5\times$ faster and can return on an average $50\%$ fewer machines per query when compared to repeated greedy {\em set cover} and baseline routing techniques.
\eat{\item We  share  a  synthetic  correlated query workload  generator  for  evaluating  our system that may be of independent interest for evaluating other systems.}
\squishend

The remainder of the paper is structured as follows.
Sections~\ref{sec:related_work} and~\ref{sec:background} present related work and problem background.\eat{, making the set cover problem more precise and describing our background research that went into tackling it} In Section~\ref{sec:clustering} we describe our query clustering algorithm. Section~\ref{sec:cluster_processing} explains how we deal with the clusters once they are created and processing of real-time queries in Section~\ref{sec:real-time}.\eat{ we used to combine similar queries, and demonstrate its effectiveness Section~4 then explains how we deal with the clusters once they are created and gives experimental results.} Finally, Section~\ref{sec:experiments} discusses the experimental evaluation of our techniques on both real-world and synthetic datasets, followed by conclusion. \eat{ Finally, in Section~6, we discuss future directions that could be taken to solve this problem more completely.} 

\section{Related Work} 
\label{sec:related_work}
SCHISM by Curino et al.,~\cite{Curino:2010:SWA:1920841.1920853} is one of the early studies in this area that primarily focuses on minimizing the number of machines a database transaction touches, thereby improving the overall system throughput. In the context of distributed information retrieval Kulkarni et al.,~\cite{Kulkarni:2015:SSE:2766484.2738035} show that minimizing the number of document shards per query can reduce the overall query cost. The above related work does not focus on speeding up query routing. Later, Quamar et al.,~\cite{Quamar:2013:SSW:2452376.2452427}\cite{Kumar:2014:SWD:2691523.2691545} presented SWORD, showing that in a scale-out data replicated system minimizing the number of machines accessed per query/job ({\em query span}) can minimize the overall energy consumption and reduce communication overhead for distributed analytical workloads. In addition, in the context of transactional workloads, they show that minimizing query span can reduce the number of distributed transactions, thereby significantly increasing throughput. In their work, however, the router executes the greedy {\em set cover} algorithm for each query in order to minimize the {\em query span}, which can become increasingly inefficient as the number of queries increases. Our work essentially complements all the above discussed efforts, with our primary goal being to improve the query routing performance while retaining the optimality by sharing the {\em set cover} computations among the queries.

There are numerous variants of the {\em set cover} problem, such as an online {\em set cover} problem~\cite{DBLP:journals/siamcomp/AlonAABN09} where algorithms get the input in streaming fashion. Another variant is, $k$-{\em set cover} problem~\cite{unweighted} where the size of each selected set does not exceed $k$. Most of the variants deal with a single universe as input~\cite{agg13}\cite{kle05}\cite{vaz13}, whereas in our work, we deal with multiple inputs (queries in our case). Our work is the first to enable sharing of {\em set cover} computations across the inputs/queries thereby improving the routing performance significantly.

In this work, in order to maximize the sharing of {\em set cover} computations across the queries, we take advantage of correlations existing between the queries. Our key approach is to cluster the queries so that queries that are highly similar belong in the same cluster. Queries are considered highly similar if they share many of their data points. By processing each cluster (instead of each query) we are able to reduce the routing computation time. There is rich literature on clustering queries to achieve various objectives. Baeza-Yates et al.,~\cite{qruqlinse} perform clustering of search engine queries to recommend topically similar queries for a given future query. In order to analyze user interests and domain-specific vocabulary, Chuang et al.,~\cite{1183888} performed hierarchical Web query clustering. There is very little work in using query clustering to speed up query routing, while minimizing the average number of machines per query for scale-out architectures. Our work provides one of the very first solutions in this space.  

Another study~\cite{2002:QCU:503104.503108} describes the search engine query clustering by analyzing the user logs where any two queries is said to be similar if they touch similar documents. Our approach follows this model, where a query is represented as a set of data items that it touches, and similarity between queries is determined by the similar data items they access.


\section{Problem Background}
\label{sec:background}
Mathematically, the set cover problem can be described as follows: given the finite universe $Q$, and a collection of sets $\mathscr{M} = \{S_1,S_2, \ldots, S_m\}$, find a sub-collection we call cover of $Q$, $\mathscr{C}\subseteq \mathscr{M}$, of minimal size, such that $Q\subseteq\bigcup\mathscr{C}$\eat{(oftentimes, when talking about covering just one subset of $U$, we identify $Q$ to be our universe and for sets take their intersections with $Q$)}. This problem is proven to be \textit{NP-hard} \cite{kle05}\eat{(there is no known efficient algorithm that solves it)}. Note that a brute force search for a minimal cover requires looking at $2^m$ possible covers. Thus instead of finding the optimal solution, approximation algorithms are used which trade optimality for efficiency~\cite{pas97}\cite{vaz13}. The most popular one uses a greedy approach where at every stage of the algorithm, choose the set that covers most of the so far uncovered part of $Q$, which is a $\ln n$ approximation that runs in 
$O(n)$ time. 

The main focus of this work is the \textit{incremental set cover problem}. 
Mathematically, the only difference from the above is 
that instead of covering only one universe $Q$, set covering is performed on each universe 
$Q_{i}$ from a collection of universes $\mathscr{Q} = \{Q_1, Q_2,\dots, Q_N\}$. 
Using the greedy approach separately on each $Q_{i}$ from 
$\mathscr{Q}$ is the na\"{i}ve approach, but when $N$ is large, running the greedy 
algorithm repeatedly becomes unfeasible. We take advantage of information about 
previous computations, storing information and using it later to compute remaining 
covers faster. \eat{Another natural constraint is 
assigning a load $l_{i}$ to each $M_{i}$ from $\mathscr{M}$ so that each $M_{i}$ 
can be used at most $l_{i}$ times when covering sets from $\mathscr{Q}$.}

In this paper, elements in the universe are called \textit{data}, sets from $\mathscr{M}$ are \textit{machines} and sets from $\mathscr{Q}$ are \textit{queries}. \eat{EXPERIMENTS --- The sizes considered in this project are the following: each data unit is replicated $3$ times, all machines are the same in size, and in our case the size of the universe is about $100000$ units of data. The numbers of queries is about $50000$, and each query asks for between $6$ and $15$ data units.} Realistically, it can be assumed that data are distributed randomly on the machines with replication factor of $r$\eat{ \cite{kum14}}\eat{, and that queries are somehow correlated}. In this work, we take advantage of the fact that, in real world, queries are strongly correlated~\cite{DBLP:conf/ijcai/ZhaoSXHZ15}\cite{DBLP:conf/sigmod/GuptaKRBGK11} and enable sharing {\em set cover} computations across the queries. Essentially, this amortizes the routing costs across the queries improving the overall performance while reducing the overall cost of analytics by minimizing the average {\em query span} for the given query workload~\cite{Kumar:2014:SWD:2691523.2691545}\cite{Quamar:2013:SSW:2452376.2452427}.\eat{The correlation arises from the practical considerations of database storage since it is unlikely that queries are completely independent of each other in a real system. For example, imagine that the different queries from the same company access similar data items or that some data items are often access together. There are two cases to consider: when queries are known beforehand, and when they arrive in real-time. In the first case, all queries are known beforehand but it is not known in what order the queries will arrive.\eat{ It should be noted that it is assumed that there is enough time and memory for pre-processing in the first case before the queries start actually arriving.} In the second case we know a fraction of the queries and process them as in the first case; then new queries will begin arriving for which we have no prior information.} \eat{When queries arrive in real-time they do not arrive continuously; there will be periods of time where there will be no queries arriving.}




\newtheorem{proposition}{Proposition}

\section{Query Clustering}
\label{sec:clustering}
In order to speedup the {\em set cover} based routing, our key idea is to reduce the number of queries needed to process. More specifically, Given $N$ queries, we want to cluster the them into $m$ groups ($m<<N$) so that we can calculate set cover for each cluster instead of calculating set cover for each query. Once we calculate set cover for each cluster, next step would be to classify each incoming real-time query to one of the clusters and re-use the pre-computed set cover solutions to speedup overall routing performance. To do so, we employ clustering as the key technique for precomputation of the queries. 
An ideal clustering algorithm would cluster queries that had large common intersections with each other; it would also be scalable since we are dealing with large numbers of queries. \eat{Especially for the real-time case,}
In order to serve real-time queries we need an incoming query to be quickly put into the correct cluster.

Most of the\eat{simpler} clustering algorithms in the literature require the number of clusters to be given by the user. 
However, we do not necessarily know the number of clusters beforehand\eat{(see
\cref{Ch:Future} for potential bounds on the number of clusters). We are also willing to accept some 
mediocre clusters since we need only to beat the optimality of the brute force strategy}. We also want to be 
able to theoretically determine bounds for the size of clusters, so our final algorithm can have bounds as well. 
To that effect, we developed entropy-based real-time clustering algorithm\eat{s heavily deriving from the literature (especially 
COOLCAT, described in \cite{bar02}), but with important modifications}. 
\label{subsec:entropy}
Using entropy for clustering has precedent in the literature (see \cite{bar02}). 
Assume that we have our universe $U$ of size $n$, let $K$ be a cluster containing queries $Q_1, \ldots, Q_m$. Then we can define the probability $p_j$ of data item 
$j$ being in the cluster $K$: 
\begin{equation}
p_j(K) = \frac{1}{|K|} \sum_{i=1}^{|K|}\chi_j(Q_i)
\end{equation}
where the characteristic function $\chi_j$ is defined by: 
\begin{equation}
\chi_j(Q) = \left\{
\begin{array}{ll}
1, & j \in Q \\
0, & j \notin Q
\end{array}
\right.
\end{equation}
Then we can define the \emph{entropy} of the cluster, 
$S(K)$ as 
\begin{equation}
S(K) = - \sum_{j=1}^n p_j(K)\log_2 p_j(K) + (1 - p_j(K))\log_2 (1- p_j(K)) \label{eq:entropy}
\end{equation}


This entropy function is useful because it peaks at $p=0.5$ and is $0$ 
at $p=0$ and $p=1$. Assume 
we are considering a query $Q$ and seeing if it should join 
cluster $K$. For any data element $j\in Q$, if most of the elements in $K$ do not contain $j$, then 
adding $Q$ to $K$ would increase the entropy; conversely if most of the elements contain $j$, then 
adding $Q$ would decrease the entropy. Thus, minimizing entropy forces a high degree of similarity between clusters.

\subsection{The $\mathtt{simpleEntropy}$ Clustering Algorithm:} \label{sec:simpleEntropy}
We developed a simple entropy-based algorithm (pseudocode shown in Algorithm~\ref{alg:simple_entropy}).\eat{which can be found in \cref{alg:simentropy}. As} As each query $Q$ comes in, we compute the entropy of placing the query in each of the current 
clusters and keep track of the cluster which minimizes the \emph{expected entropy}:
given clusters $K_1, \ldots K_m$ in a clustering $\mathscr{K}$, 
the expected entropy is given by: 
\begin{equation}
\label{eq:expent}
\mathbb{E}(\mathscr{K}) = \frac{1}{m}\sum_{j=1}^{m}|K_j|\cdot S(K_j) 
\end{equation}
If this 
entropy is above the \emph{threshold} described below, the query starts its own 
cluster. Otherwise, the query is placed into the cluster that minimizes entropy. 

Suppose we are trying to decide if query 
$Q = \{x_1, \ldots, x_n\}$ 
should be put into cluster $K$. Let $p_i$ be the frequency with which 
$x_i$ is in the clusters of $K$. Then define the set
\[
T(Q,K) = \{x_i \in Q \: : \: p_i \ge \theta_1\}
\]
for some threshold $\theta_1$. We say that $Q$ is \emph{eligible} for 
placement in $C$ if $|T(Q,C)| \ge \theta_2 |Q|$ for some other threshold
$\theta_2$. Essentially, we say that $Q$ is eligible for placement in $K$ only if ``most of the elements in 
$Q$ are common in $K$,'' where ``most'' and ``common'' correspond to $\theta_1$ and $\theta_2$ and are user-defined.
Of course, we should have $0 \le \theta_1, \theta_2 \le 1$. Then, 
given a clustering $\mathscr{K}$ with clusters $K_1, \ldots, K_m$, we create 
a new cluster for a query $Q$ only when $Q$ is not eligible for placement into
any of the $K_i$. This forces most of the values in 
the query to `agree' with the general structure of the cluster. 

The goal is an algorithm that generates clusters with low entropy. Let us say that a \emph{low-entropy cluster}, 
a cluster for which more than half the data elements contained in it have probability at least $0.9$, is a \emph{tight} cluster. The
opposite is a $\emph{loose}$ cluster, i.e. many elements have probability close 
to $0.5$.


\begin{algorithm}
\caption{A simple entropy based real-time clustering algorithm. Here $\mathscr{Q}$ is the list of 
queries, and $\theta_1$ are $\theta_2$ are the threshold parameters defined in
Section~\ref{subsec:entropy}}
\begin{algorithmic}

\Function{simpleEntropy}{$\mathscr{Q}$, $\theta_1$, $\theta_2$}
\State $\mathscr{C} \gets \varnothing$
\Comment $\mathscr{C}$ is the set of clusters
\State $S \gets 0$ 
\Comment $S$ holds the current expected entropy
\State $M \gets 0$
\Comment $M$ is the number of clustered elements

\For{$Q \in \mathscr{Q}$}
\State $S_Q \gets \infty$
\Comment $S_Q$ is the min expected entropy due to $Q$
\State $C_Q \gets \varnothing$

\For{$C \in \mathscr{C}$}

\If{\text{NOT} \Call{Eligible}{$Q$, $C$}}
\State \text{Continue.}
\EndIf

\If{\Call{ExpectedEntropy}{$Q$,$C$} $<$ $S_Q$}
\State $S_Q \gets$ \Call{ExpectedEntropy}{$Q$,$C$}
\State $C_Q \gets C$
\EndIf

\EndFor

\If{$S_Q < \infty$}
\State \Call{AddToCluster}{$Q$,$C_Q$}
\Else
\State \Call{NewCluster}{$Q$}
\EndIf
\EndFor
\State\Return $\mathscr{C}$

\EndFunction
\end{algorithmic}
\label{alg:simple_entropy}
\end{algorithm}

\subsection{Analysis of the $\mathtt{simpleEntropy}$ Clusters:} We take a more in-depth look at the type of clusters that form with an entropy clustering algorithm. The first question considered was whether the algorithm is more likely to generate one 
large cluster or many smaller clusters. To do this, we considered how the size of a cluster affects 
whether a new query is added to it. We are also interested in how the algorithm
weights good and bad fits: If a query $Q$ contains data elements 
$\{y, x_1, \ldots x_n\}$, we want to determine how many of the $x_i$ need to 
be common in a cluster $K$ to outweigh $y$ being uncommon in $K$. 

The setup is as follows. Assume that we have already sorted $m$ queries into a clustering
$\mathscr{K}$, and assume the
expected entropy of this clustering is $\Omega$. Given $m$ clusters 
$\{K_1, \ldots, K_m\}$, where each cluster $K_j$ has entropy $E(K_j)$, the 
expected entropy is given in Equation~\ref{eq:expent}\eat{\eqref{eq:expent}}. 

Now, we want to calculate the change in expected entropy when a new query, $Q$, is added to a cluster, 
as a function of the cluster's size and composition. As a simple start, we only consider the 
change due to a single data entry. Since the total entropy is additive, 
understanding the behavior due to one data entry helps understand where 
the query is allocated. 

Let $\Delta\mathbb{E}_i$ be the change in expected entropy due to the 
presence or absence of data element $i$. Let $p_i$ be the 
probability value for element $i$ in a given cluster, and let $p^*_i$ be the new probability if the 
query is added to the cluster. We have: 
\begin{equation}
p^*_i = \left\{
\begin{array}{ll}
\dfrac{np_i+1}{n+1}, & i \in Q \\ \\
\dfrac{np_i}{n+1}, & i \notin Q
\end{array}
\right.
\label{eq:pstar}
\end{equation}
Let us also define $S(p_i)$ as the entropy of a single data element, i.e.
\begin{equation}
S(p_i) = -p_i \log_2 p_i - (1-p_i)\log_2(1-p_i)
\end{equation}

\begin{proposition}
With the above pre-requisites, we can derive that the difference in expected entropy due to 
data element $i$ by adding a query to a cluster of size $n$ which had a probability $p_i$ for element
$i$, is: 
\begin{equation}\label{eq:deltaent}
\footnotesize
\Delta\mathbb{E}_i(\mathscr{K}) = \frac{1}{M+1}(M\Omega - nS(p_i) + (n+1)S(p^*_i)) - \Omega
\end{equation}
\end{proposition}
\begin{proof}
The derivation is as follows: To get the new expected entropy, we need to remove the old weighted 
entropy of the cluster, which is given by $nS(p_i)$ and add back the new weighted entropy of that 
cluster, given by $(n+1)S(p_i^*)$. And now there are $M+1$ total elements so we divide by 
$M+1$ to obtain the new expected entropy. Then to get the difference, we simply subtract the old 
expected entropy, $\Omega$. 
\end{proof}



\begin{figure}[tb]
\centering
\subfigure[$i\in Q$]{ \includegraphics[width=.22\textwidth]{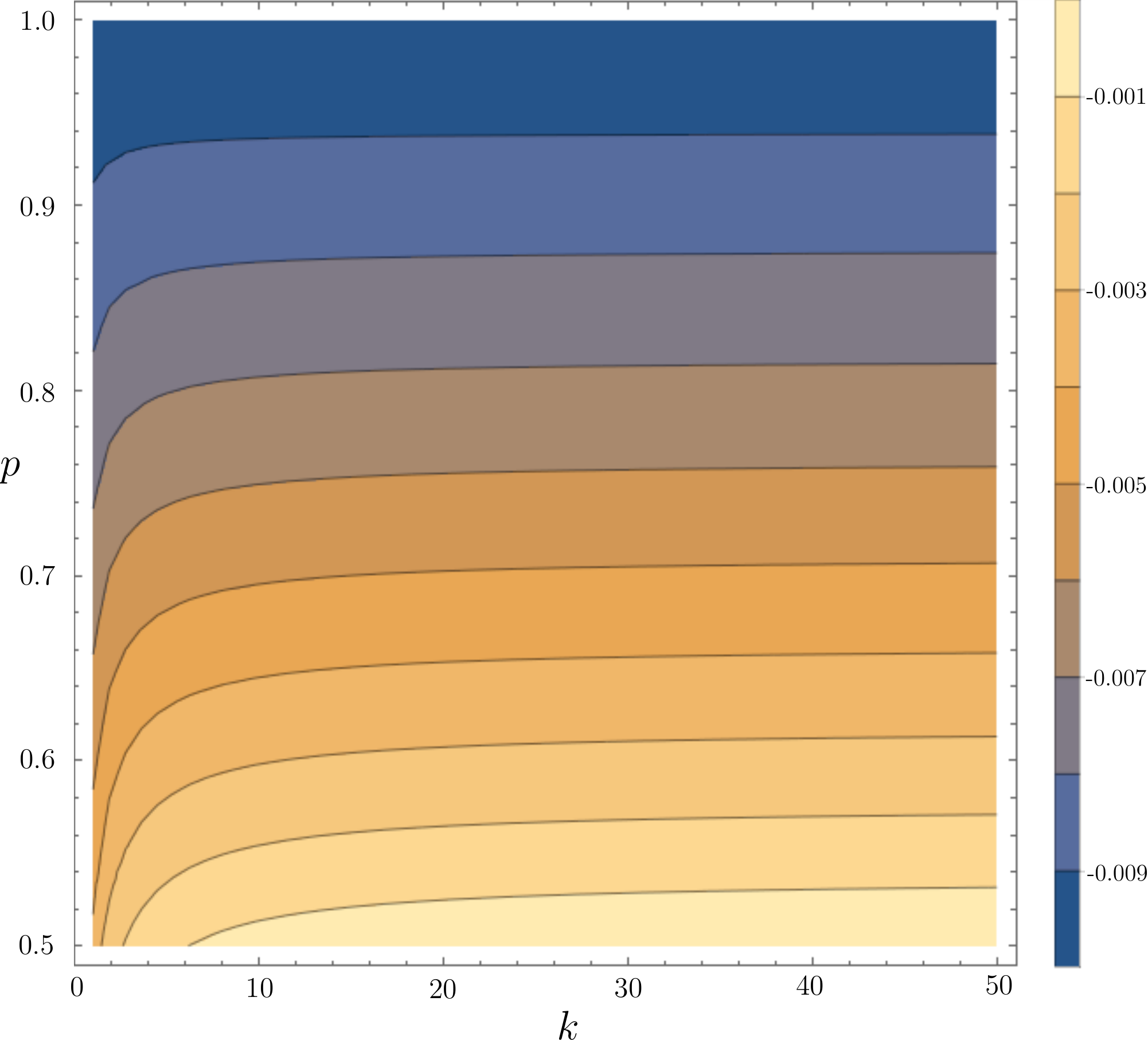} \label{fig:half_in}}
\hfill
\subfigure[$i\notin Q$]{\includegraphics[width=.22\textwidth]{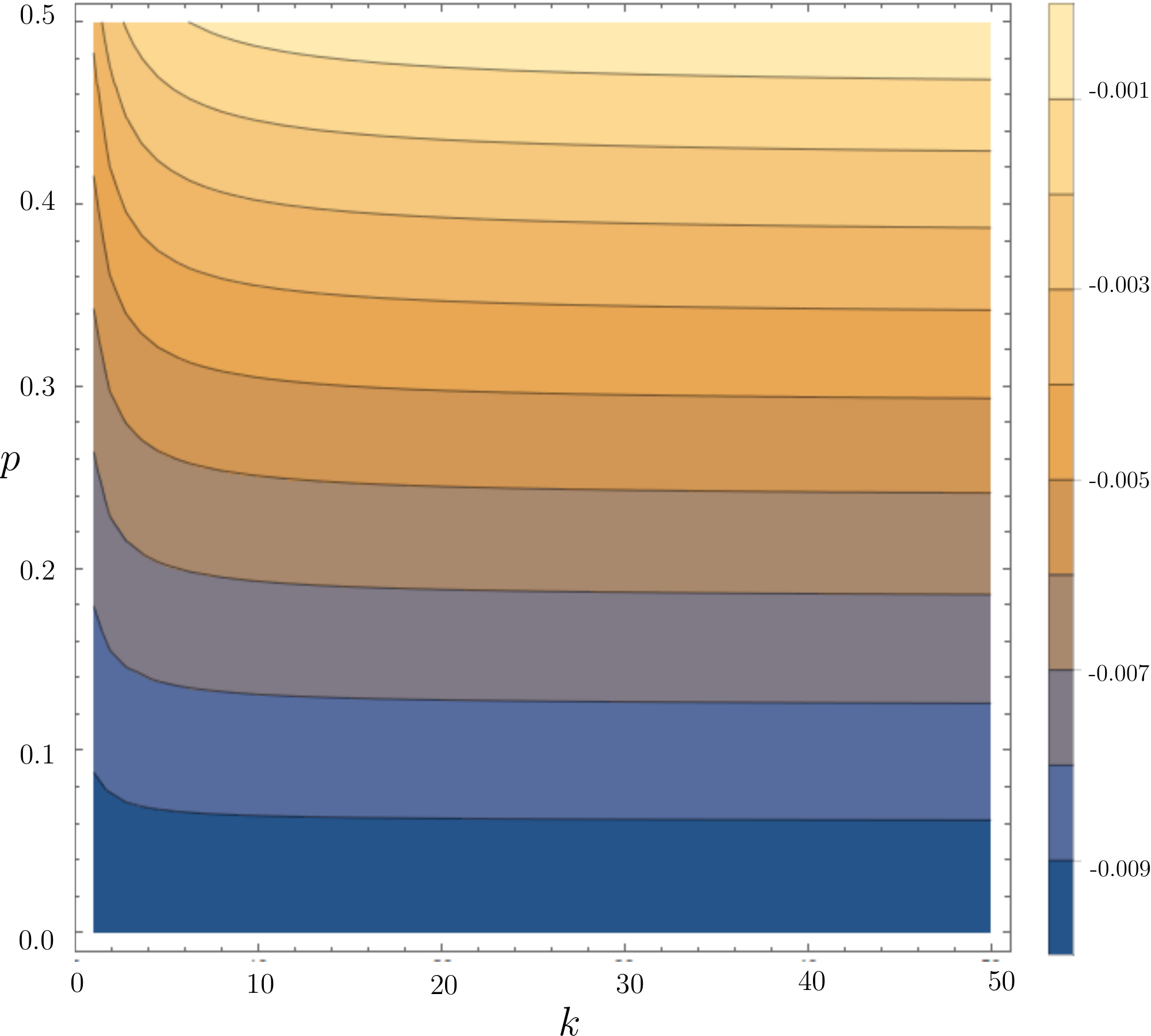} \label{fig:half_out}}
\caption[Expected entropy for single data elements]{Here we plot Equation~\ref{eq:deltaent} setting the expected entropy $\Omega = 1$ and the total number of queries processed
$M = 100$, but the general shapes remain true for a wide variety of parameters. 
The $x$-axis is the number of queries in the cluster under consideration and the 
$y$-axis is the probability of element $i$ in that cluster. We see that the 
entropy difference is independent of cluster size for large enough clusters.}
\label{fig:half}
\end{figure}

Figure~\ref{fig:half} shows some results of this analysis. 
As clusters become large enough, $\Delta\mathbb{E}$ (the change in 
expected entropy), becomes constant
with respect to the size of the cluster, and change steadily with respect to $p$. 
This is desirable because it means that the change in expected entropy caused by 
adding $Q$ is dependent only on $p_i$ and not on the size of the cluster. 

As will be shown in Section~\ref{Ch:QueryProcessing} it is important to our cluster processing scheme that all the queries 
in a cluster share a large common intersection. That is, there should be data elements that are contained in 
all queries; and when a new query arrives, that query should contain 
all those shared elements. This property 
can be given a practical rationale: the data elements present in the large common intersection are probably 
highly correlated with each other (i.e. companies that ask for one of the elements probably ask for them all), so 
it makes sense that an incoming query would also contain all of those data elements if it contained any of them. 

From a theoretical perspective, we can see that there is a high entropy penalty when an incoming query does not 
contain the high probability elements of a cluster. 

\begin{proposition}
Assume we have a cluster $K$ of $m$ 
data elements, all of which have probability $p$. A query $Q$ is being tested for fit in $K$. Let us say that $Q$ contains all but $km$ elements of the cluster $K$. Then, the entropy of adding $Q$ to cluster $K$ is: 
{\small
\begin{align}
\Delta\mathbb{E}(\mathscr{K}) &=& \frac{1}{M+1}(M\Omega - nm E(p) + (n+1)kmE\left(\frac{pn}{n+1}\right) \nonumber\\ 
& & +(n+1)(1-k)mE\left(\frac{pn+1}{n+1}\right))-\Omega 
\label{eq:deltaent_multi}
\end{align}
}
where, as in \eqref{eq:deltaent}, $\mathscr{K}$ is the clustering, 
$M$ is the number of data elements processed,
$\Omega$ is the previous expected entropy of the cluster, and $n$ is 
the number of queries in the cluster. 
\end{proposition}
\begin{proof}
The derivation of the above is also analogous to that of 
\eqref{eq:deltaent}. The total entropy of the old clustering is 
given by $M\Omega$, from which we subtract out the old weighted entropy
of cluster $K$, given by $n\cdot mE(p)$, since $E(p)$ is the entropy 
contributed by each of the $m$ data elements of $K$, and we weight this by
$n$, the size of the cluster. Then we add the new weighted entropy of 
$K$ after $Q$ is added to it: There are now $km$ data elements which 
have probability $pn/(n+1)$ (since $Q$ does not contain these); and 
there are $(1-k)m$ data elements which have probability 
$(pn+1)/(n+1)$ (since $Q$ does contain these). Finally, since we are 
looking for the change in expected entropy, we subtract out $\Omega$, 
the old expected entropy. 
\end{proof}

\begin{figure}[tb]
\centering
\label{fig:de_full}
\subfigure[\eqref{eq:deltaent_multi} for the full range of probabilties and disjointedness.]{ \includegraphics[width=.22\textwidth]{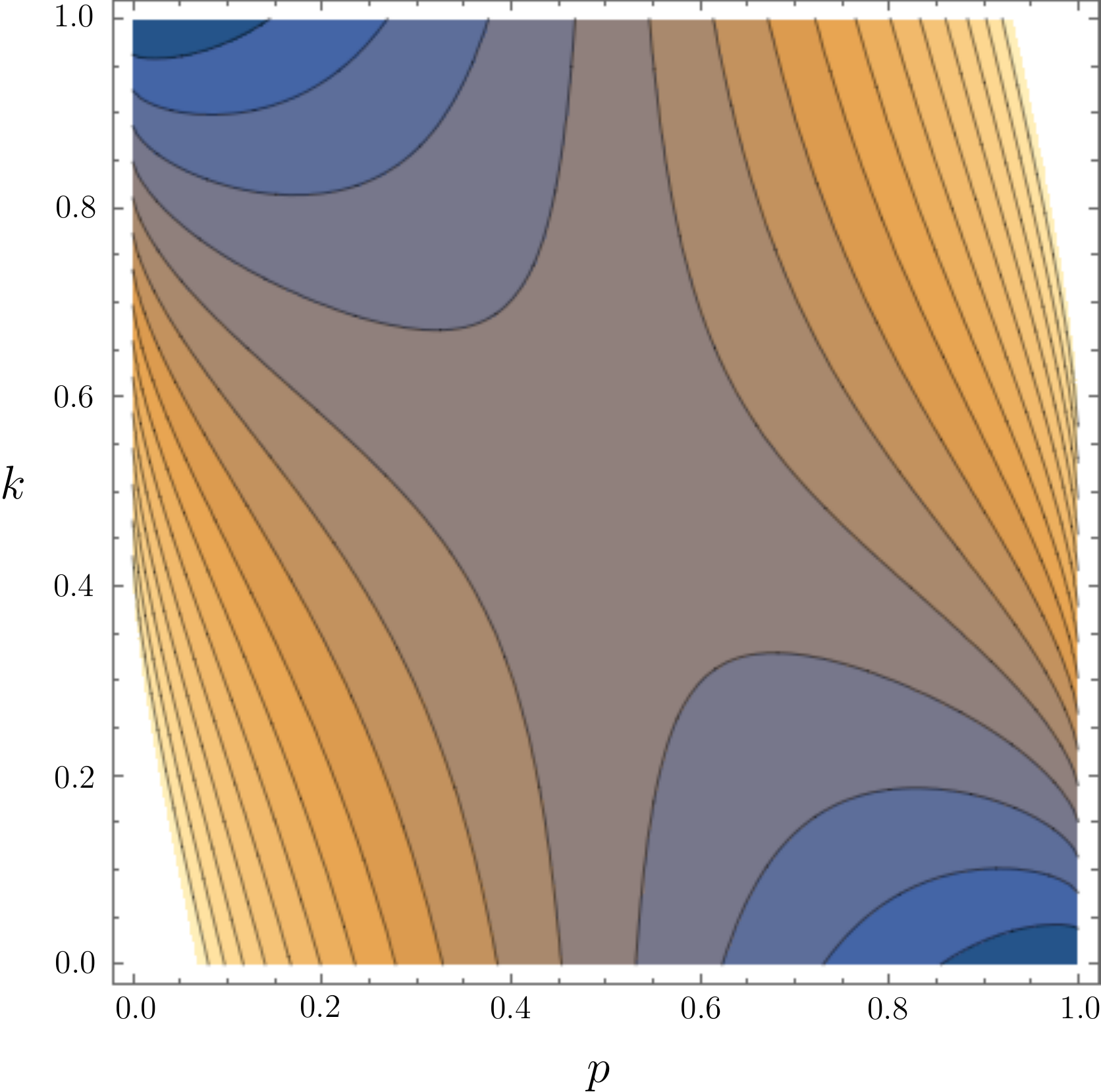}\label{fig:de_full}}
\hfill
\subfigure[\eqref{eq:deltaent_multi} concentrating on high probabilities and small $k$]{ \includegraphics[width=.22\textwidth]{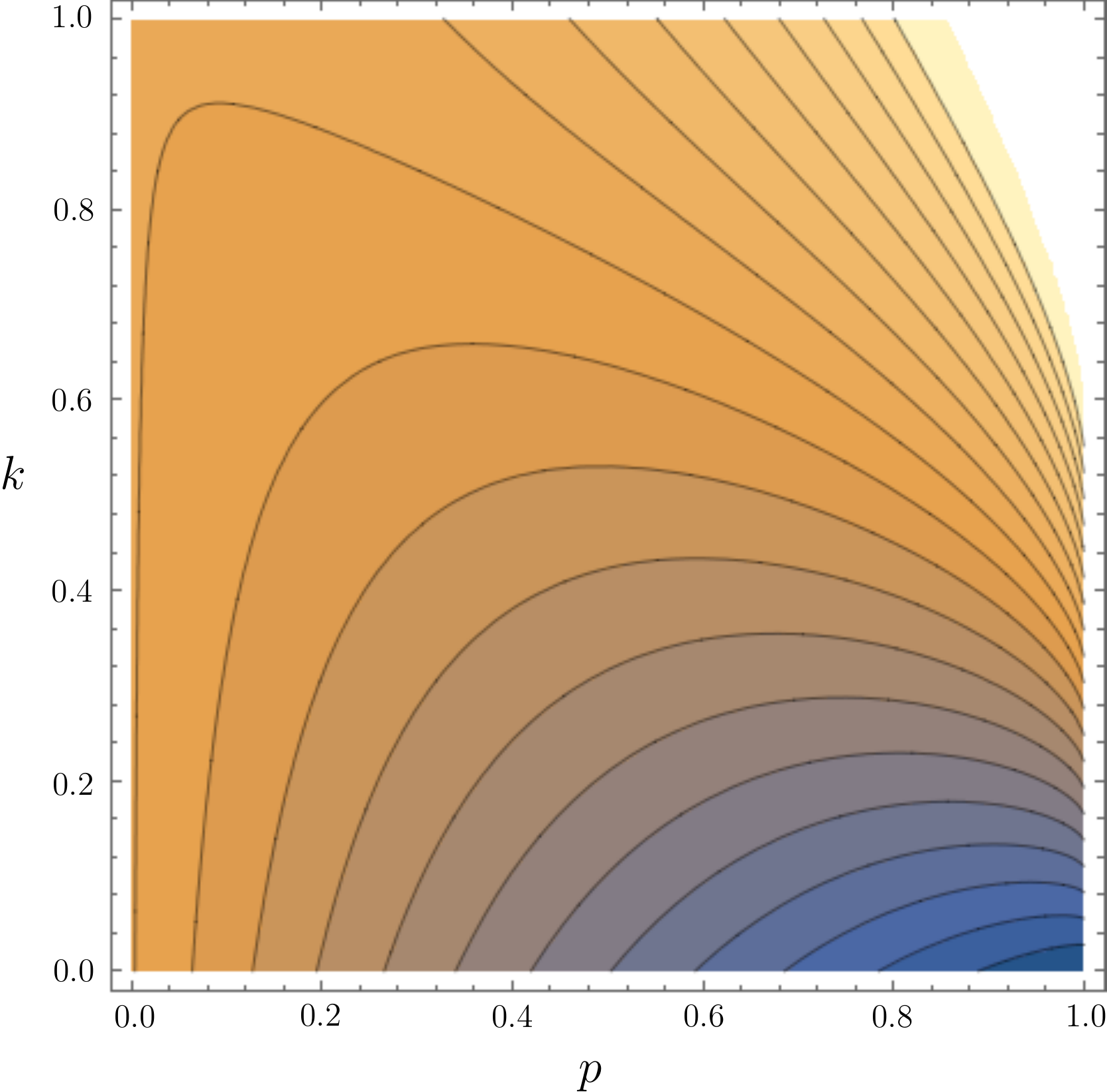}\label{fig:de_half}}
\caption[Entropy of multiple disjoint elements]{A summary of the 
theoretical basis for the effectiveness of $\mathtt{simpleEntropy}$. In 
\protect\subref{fig:de_full} we see that all contours become steeper
as we reach probability extremes. The central area is a zone of 
mediocre clustering. In \protect\subref{fig:de_half} we zoom in on 
the high probability region. Here we see that most high probability 
data items \emph{must} be contained in an incoming query.}
\label{fig:deltaent_multi}
\end{figure}

In Figure~\ref{fig:deltaent_multi}, we plot Equation~\ref{eq:deltaent_multi} holding 
$n$ and $m$ at $20$ (a typical value for clusters with our parameters according
to our experiments.)
Figure~\ref{fig:de_full} actually provides a general overview of the 
clustering landscape for an incoming query. We can divide the graph into
three regions of interest: (1) the bottom-right and top-left ($k \approx 1 - p$), or the 
\emph{high-quality} region; (2) the central region 
($k \approx p \approx 0.5$), or the \emph{mediocre} region
and 
(3) the top-right and bottom-left ($k \approx p \approx 1 \text{ or } 0$), the \emph{low-quality} region. We discuss the implication of each region in turn:

\topic{High-quality Region:}
\label{top:cluster}
These are the points of lowest entropy in the plot. In the 
bottom-right of Figure~\ref{fig:de_full}, the implication is that the cluster is a tight cluster
to begin with (most elements in it have high probability), and the
incoming query also contains most of these elements. The heavy penalty
for not containing these high probability elements reinforces the 
tightness of the cluster. We zoom in on this scenario in 
Figure~\ref{fig:de_half}, as it is the most important scenario for tight
clusters. Specifically, imagine a cluster has many elements of probability $p$, where $p$ is close to $1$. Call 
this set of elements $X$. Now, when a new query $Q$ comes in, if $X\not\subset Q$ there is a high entropy cost to\ putting $Q$ in the cluster. This means that the high probability regions of a cluster $K$ are likely conserved 
as new queries come in. As will be seen, this scenario is highly desirable for the real-time 
case. 

The scenario at the top-left of Figure~\ref{fig:de_full} is this: the set of data elements that are under consideration are not contained in 
a given cluster $K$, and the incoming query also does not contain most of these elements. This case is also important 
for describing clustering dynamics. A cluster that has been built already may contain some queries that have a few 
data elements with low probability. Let $X$ be such a set of data elements with $p$ close to $0$. Then a new query 
coming in is \emph{penalized} for containing \emph{any} of those elements. Those elements will remain with low
probability in the cluster. The implication is that once a cluster has some low probability elements, these 
elements will remain with low probability. 

\topic{Mediocre Region:}
This is potentially the most problematic region of the entire clustering. Consider a cluster where most of 
elements have probability $p$ close to $0.5$. Then when a new query comes in, the entropy change is similar 
regardless of how many elements of the cluster the query contains. And if the query is put into this cluster, then 
on average the probabilities should remain close to $0.5$, so the cluster type is perpetuated. Such a cluster is 
not very tight. 

\topic{Low-quality Region:}
This region describes when an ill-fitting query is put into a tight cluster. For example, if we have a cluster for which 
all the elements have probability $p$ close to $1$ and a query comes in which contains none of these elements, we 
would be in the top-right region of Figure~\ref{fig:de_full}. This is clearly the worst possible situation in terms of a 
clustering. Therefore it is a desirable feature of our clustering algorithm that the entropy cost for this 
scenario is relatively high.

From analyzing these regions, we see that our algorithm will generate high-quality and mediocre clusters (from the 
high-quality and mediocre regions respectively), but the entropy cost of the low-quality region is too high for 
many low-quality clusters to form. 

\section{Cluster Processing}
\label{Ch:QueryProcessing}
\label{sec:cluster_processing}
Once our queries are clustered, the goal is to effectively process the clusters as 
a whole instead of processing each query individually. To that end, we first 
introduce our so-called $\mathtt{BetterGreedy}$ algorithm, which is a modified 
version of the standard greedy algorithm more suited to this problem. 
\subsection{The $\mathtt{BetterGreedy}$ algorithm} 
Recall that the standard greedy algorithm covers a query $Q$ with a small 
number of machines. The $\mathtt{BetterGreedy}$ algorithm is performed on a query 
$Q_1$ \emph{with respect to} another query $Q_2$. The pseudocode is given in 
Algorithm~\ref{alg:better_greedy}. At stage $k$, let $Q_k\subset Q_1$ be the still 
uncovered elements of $Q_1$. We choose the machine $M^*$ that contains the most 
elements of $Q_k$. In the standard greedy algorithm, if there is a tie, an $M^*$ 
is chosen arbitrarily. In $\mathtt{BetterGreedy}$, if there is a tie, 
we choose $M^*$ so that it \emph{also} maximizes the elements covered in $Q_2$. 
See Figure~\ref{fig:subset} for a visual example.

\begin{algorithm}
\footnotesize
\caption{The $\mathtt{BetterGreedy}$ algorithm. We cover query $Q_1$ with 
respect to $Q_2$. $\mathscr{D}$ is a dictionary, where 
$\mathscr{D}[i]$ returns the list of machines whose intersection with the 
uncovered portion of $Q_1$ is size $i$, and the list is sorted by intersection 
size with $Q_2\setminus Q_1$. 
}
\label{alg:better_greedy}
\begin{algorithmic}
\Function{BetterGreedy}{$Q_1$, $Q_2$, $\mathscr{D}$}
\State $\mathscr{M} \gets \varnothing$
\Comment $\mathscr{M}$ is the list of machines
\State $Q \gets Q_1$ 
\Comment $Q$ is the uncovered portion of $Q_1$. 
\While{$|Q| > 0$}
\State $M \gets \mathscr{D}$.\Call{Max}{}
\Comment The machine with the most uncovered elements
\State $\mathscr{M}$.\Call{Add}{$M$}
\State $Q \gets Q\setminus M$
\State $\mathscr{D}$.\Call{Update}{$Q$}
\Comment Update the intersection sizes with $Q$ 
\EndWhile
\State\Return $\mathscr{M}$
\EndFunction
\end{algorithmic}
\end{algorithm}

\subsection{Analysis of the $\mathtt{BetterGreedy}$ Algorithm}
\label{sec:better_greedy}
To make the algorithm as fast as possible, we have a dictionary of lists called 
$\mathtt{sets\_of\_size}$. Each key in this dictionary is the size of the intersection of 
each machine with the current uncovered set (and this dictionary is updated at 
each stage). Corresponding value is a set of all machines of that size. \eat{When there is a standard implementation of the greedy set cover algorithm.} 
If there are multiple machines under the same key (i.e. they have the 
same size with respect to the uncovered elements of $Q_1$), greedy set cover algorithm breaks tie by choosing random machine within a particular size key. However, in the case of our $\mathtt{BetterGreedy}$ algorithm, they are sorted 
according to the size of their intersection with $Q_2\setminus Q_1$. While 
this additional sorting makes the algorithm worse than standard greedy approach in the worst case 
(since all the machines could be the same size in $Q_1$), in practice, our 
clustering strives to make $Q_2 \setminus Q_1$ small, and so the algorithm is 
fast enough. 

\begin{proposition}
\label{lineargreedyproof}
The above described algorithm runs in $O(\sum_{k=1}^m |M_{k}|) = O(r\cdot |Q|)$, where $r$ is the replication factor when data is distributed evenly on the machines.
\end{proposition}

\begin{proof}
There are two main branches of our algorithm: either under the selected key in $sets\_of\_size$ is an empty or nonempty set. In the first case, we move our counter to one key below. There are at most $|Q|$ keys in this dictionary, so this branch, which we call a \textit{blank step}, will be performed at most $O(|Q|)$ times, and since it takes $O(1)$, the total time for this branch is $O(|Q|)$. In the other branch, when the set under selected key is not an empty set, we see in the description of the algorithm that in the innermost loop a data unit from a machine is removed (and some other things that all take $O(1)$ are performed). Since there are $\sum_{k=1}^m |M_{k}|$ data units in all machines combined, this part runs in $O(\sum_{k=1}^m |M_{k}|)$. To conclude, the whole algorithm runs in $O(\sum_{k=1}^m |M_{k}|+|Q|)=O(\sum_{k=1}^m |M_{k}|)$.
\end{proof}

\subsection{Processing Simple Clusters}

In this section we describe the most basic clusters and our ways to process them. 

\begin{figure*}
\centering
\begin{minipage}{0.48\textwidth}
\centering
\subfigure[Illustration of $\mathtt{BetterGreedy}$]{
\includegraphics[width=0.8\textwidth]{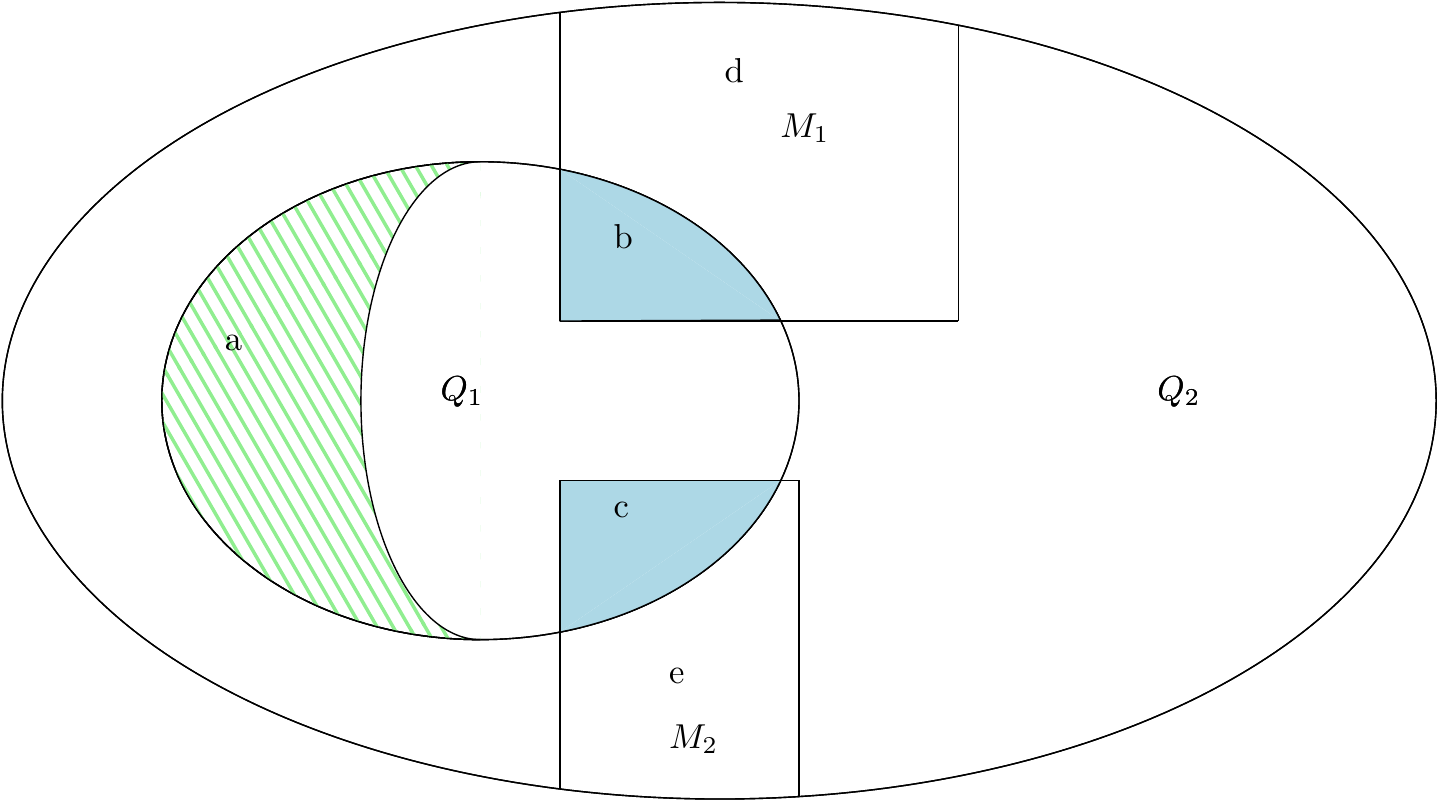}
\label{fig:subset}
}
\end{minipage}
\begin{minipage}{0.48\textwidth}
\centering
\subfigure[Visual representation of the query intersection algorithm]{\includegraphics[width=0.6\textwidth]{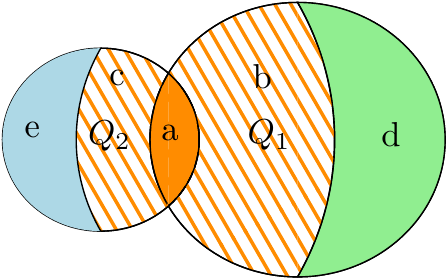}
\label{fig:greedy_intersection}}
\end{minipage}
\caption[]{In \textbf{\protect\subref{fig:subset}}, we see an example of $\mathtt{BetterGreedy}$. Assume region \textbf{a} has already been covered, and we have two machines $M_1$ and $M_2$ that cover regions \textbf{b} and \textbf{c} of $Q_1$ respectively, and let these regions be the same size. While the standard greedy algorithm would pick from $M_1$ and $M_2$ randomly, $\mathtt{BetterGreedy}$ chooses $M_1$ because the size of region \textbf{d} is bigger than the size of region \textbf{e}. The representation of $\mathtt{BetterGreedy}$ in \textbf{\protect\subref{fig:greedy_intersection}} demonstrates that we run the algorithm on the intersection (region \textbf{a}) and the striped sections (regions \textbf{b} and \textbf{c}), thus covering $Q_1$ and $Q_2$. Then we simply run the greedy algorithm on the remaining uncovered sections (regions \textbf{d} and \textbf{e}) to get the full covering. So, for example, the covering of $Q_1$ is given by the coverings of regions \textbf{a}, \textbf{b}, and \textbf{d}. With standard parameters, we find that covers are on average $0.15$ machines larger than the greedy cover.}
\end{figure*}

\eat{\begin{figure}[tb]
\includegraphics[width=.45\textwidth]{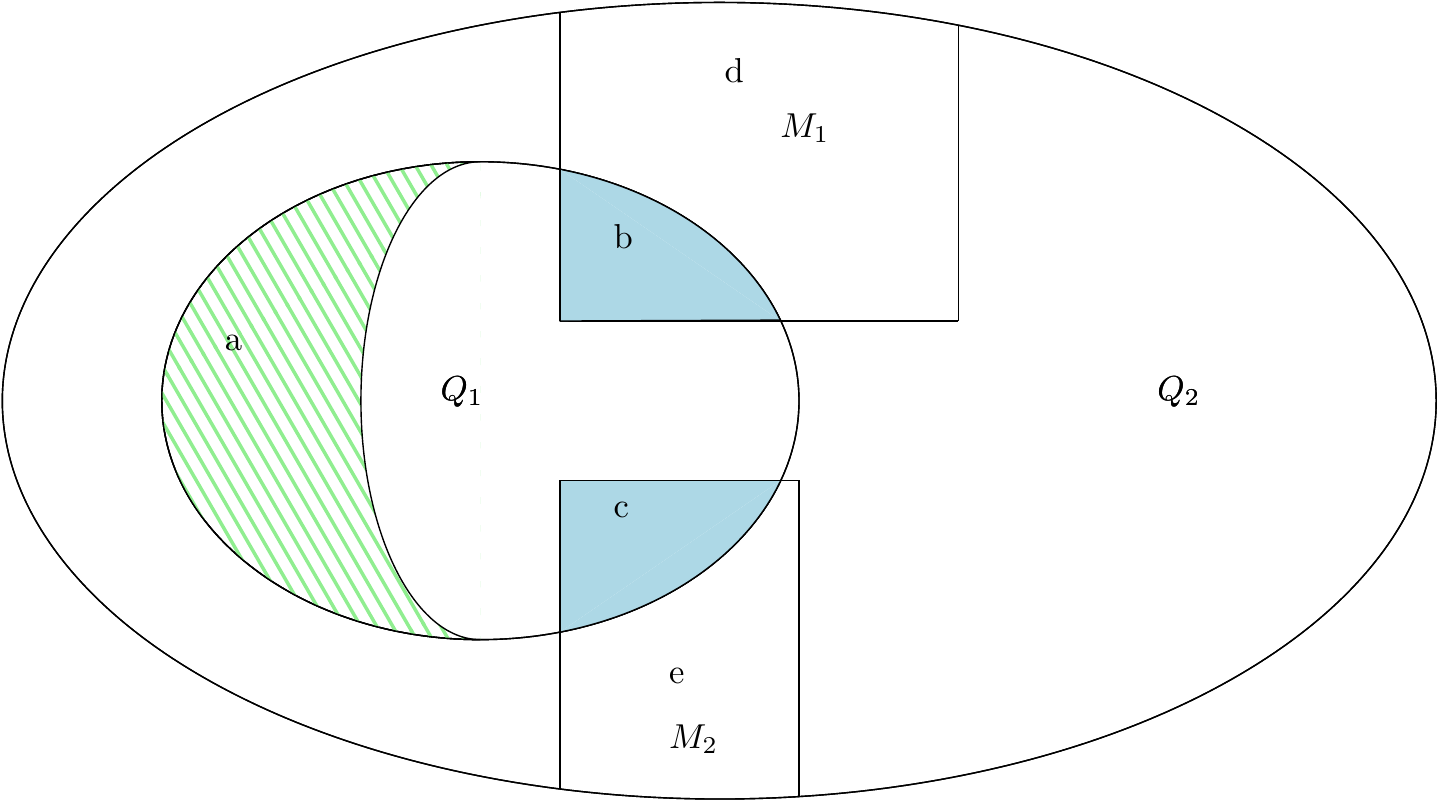}
\caption[Illustration of $\mathtt{BetterGreedy}$]{Here we see an example of $\mathtt{BetterGreedy}$. Assume region \textbf{a} has already been covered, and we have two machines $M_1$ and $M_2$ that cover regions \textbf{b} and \textbf{c} of $Q_1$ respectively, and let these regions be the same size. While the standard greedy algorithm would pick from $M_1$ and $M_2$ randomly, $\mathtt{BetterGreedy}$ chooses $M_1$ because the size of region \textbf{d} is bigger than the size of region \textbf{e}.
\label{fig:subset}}
\end{figure}}

 \begin{table}[tb]
 \footnotesize
 \centering
 \begin{tabular}{l l l l}
 \hline
 Algorithm Type & $Q_1$ & $Q_2$ &Uncovered Part\\ 
 \hline
 Cover just $Q_2$ & 1.06+ $G_1$ & $G_2$ \\ 
 Cover $Q_1$ with greedy & $G_1$ & $G_2$ + .41 & 1.80 \\
 Cover $Q_1$ with $\mathtt{BetterGreedy}$ & $G_1$ & $G_2$ + .05 & 1.31 \\ 
 \hline
 \end{tabular}
 \vspace{3pt}
 \caption[Processing simple clusters]{A comparison of the three methods for processing the cluster $Q_1 \subset Q_2$. $G_1$ gives the 
 size of the greedy cover of $Q_1$, and $G_2$ gives the size of the greedy cover of $Q_2$. Note that using $\mathtt{BetterGreedy}$
 to process $Q_1$ gives us a solution for $Q_2$ that is nearly always as good as the greedy solution.}
 \label{tab:bg_v_g}
 \end{table}

\topic{Nested Queries:} Consider the most simple query cluster: just two queries, $Q_1$ and $Q_2$, such that $Q_1 \subset Q_2$. One might suggest to simply find a cover only for $Q_2$, using the greedy algorithm, and use it as a cover for both $Q_1$ and $Q_2$. In practice, this approach does not perform well. This algorithm is unacceptable in terms of optimality of the cover when comparing the size of the cover for $Q_1$ given by this approach to the size of the cover produced by the greed algorithm (see Table~\ref{tab:bg_v_g}). Figure~\ref{fig:subset} with caption explains how our approach solves the problem judiciously.

\topic{Intersecting Queries:} Here we consider a simple cluster with two queries, $Q_1$ and $Q_2$ such that $Q_1\cap Q_2 \neq \emptyset$. We will do the following: 1) Cover $Q_1\cap Q_2$ using $\mathtt{BetterGreedy}$ with respect to $Q_1\cup Q_2$ 2) Cover the uncovered part of $Q_1$ and $Q_2$ separately, using the standard greedy algorithm 3) For the cover of $Q_1$ return the union of the cover for intersection and uncovered part of $Q_1$, and for the cover of $Q_2$ we will give the union of the cover for the intersection and the uncovered part of $Q_2$. By doing so, we run the $\mathtt{BetterGreedy}$ once and greedy algorithm twice instead of just running the greedy algorithm twice. However, in the first case, those two greedy algorithms and the $\mathtt{BetterGreedy}$ algorithm are performed on a smaller total size than two greedy algorithms in the second case. Our algorithm never processes the same data point twice, while the obvious greedy algorithm on $Q_1$ and $Q_2$ does. In terms of optimality of the covers obtained in this way, they are on average $0.15$ machines (each) larger than the covers we would get using the greedy algorithm. Figure~\ref{fig:greedy_intersection} gives a visual representation of the algorithm.

\eat{\begin{enumerate}
\setlength{\itemsep}{0pt}
\setlength{\parsep}{3pt}
\setlength{\topsep}{3pt}
\setlength{\partopsep}{0pt}
\setlength{\labelwidth}{1em}
\setlength{\labelsep}{0.5em}
\item Cover $Q_1\cap Q_2$ using $\mathtt{BetterGreedy}$ with respect to $Q_1\cup Q_2$ 
\item Cover the uncovered part of $Q_1$ and $Q_2$ separately, using the standard greedy algorithm 
\item For the cover of $Q_1$ return the union of the cover for intersection and uncovered part of $Q_1$, and for the cover of $Q_2$ we will give the union of the cover for the intersection and the uncovered part of $Q_2$
\end{enumerate}}

\subsection{The General Cluster Processing Algorithm (GCPA)}
\label{sec:GCPA}
Using ideas from the previous sections, we developed an algorithm for processing any cluster. We call it the General Cluster Processing Algorithm ($\mathtt{GCPA}$). The algorithm, in the simplest terms, goes as follows:
\eat{\begin{enumerate}\item }
1) Assign a value we call \emph{depth} to each data unit appearing in queries in our cluster. The depth of a data element is the number of queries that data unit is in. For example, consider a visual representation of a cluster on Figure~\ref{fig:mq_processing}(a). On the same figure under Figure~\ref{fig:mq_processing}(b) shows depths of different parts of the cluster. 
2) Divide all data units into sets we call \textit{data parts} according to the following rule: two data items are in the same data part if and only if they are contained in exactly the same queries. This will partition the union of the cluster. Also, we keep track of which parts make up each query (which we store in a hash table). 
3) We cover the data parts with our desired algorithm (greedy, $\mathtt{BetterGreedy}$, \dots) 
4) For each query we return the union of covers of data parts that make up that query as its cover.
\eat{\end{enumerate}}
 
This algorithm can process any shape of a given cluster and allows for a choice of the algorithm used to cover separate data parts. The big advantage of this algorithm is that each data unit (or data part) will be \emph{processed only once}, 
instead of multiple times as it would be if we were to use the greedy algorithm on each query separately. While dividing the cluster up into its constituent data parts is intensive, 
this is all pre-computing, and can be done at anytime once the queries are known. Again, it is important to note that our algorithms rely on being able to perform pre-computing for them to be effective. 

Since $\mathtt{BetterGreedy}$ chooses machines that cover as many elements in the 
cluster as a whole as possible, 
the covers of the data parts overlap, and makes their union smaller. 
One thing that can also be used in our favor is that when covering a certain part, 
we might actually cover some pieces of parts of smaller depths, as illustrated in 
Figure~\ref{fig:mq_processing}(c). Then, instead of covering the whole of those parts, 
we 
can cover just the uncovered elements. Figure~\ref{fig:mq_processing} shows how this work
s step by step\eat{ (note that the cluster shown is not a representative one and th
at certain details are omitted in our discussion for the sake of clearness; see app
endix for those details)}. This version of $\mathtt{GCPA}$, in which we use the gre
edy algorithm, we call $\mathtt{GCPA\_G}$.

Another option would be to use $\mathtt{BetterGreedy}$ for processing the parts. 
The $\mathtt{BetterGreedy}$ algorithms is done on the data parts with respect to 
the union of all queries containing that data part and is called 
$\mathtt{GCPA\_BG}$.
As we will see in the next section, this algorithm gives a major improvement in the
optimality of the covers compared to $\mathtt{GCPA\_G}$.

\begin{figure}[tb]
\centering
\includegraphics[width=.45\textwidth]{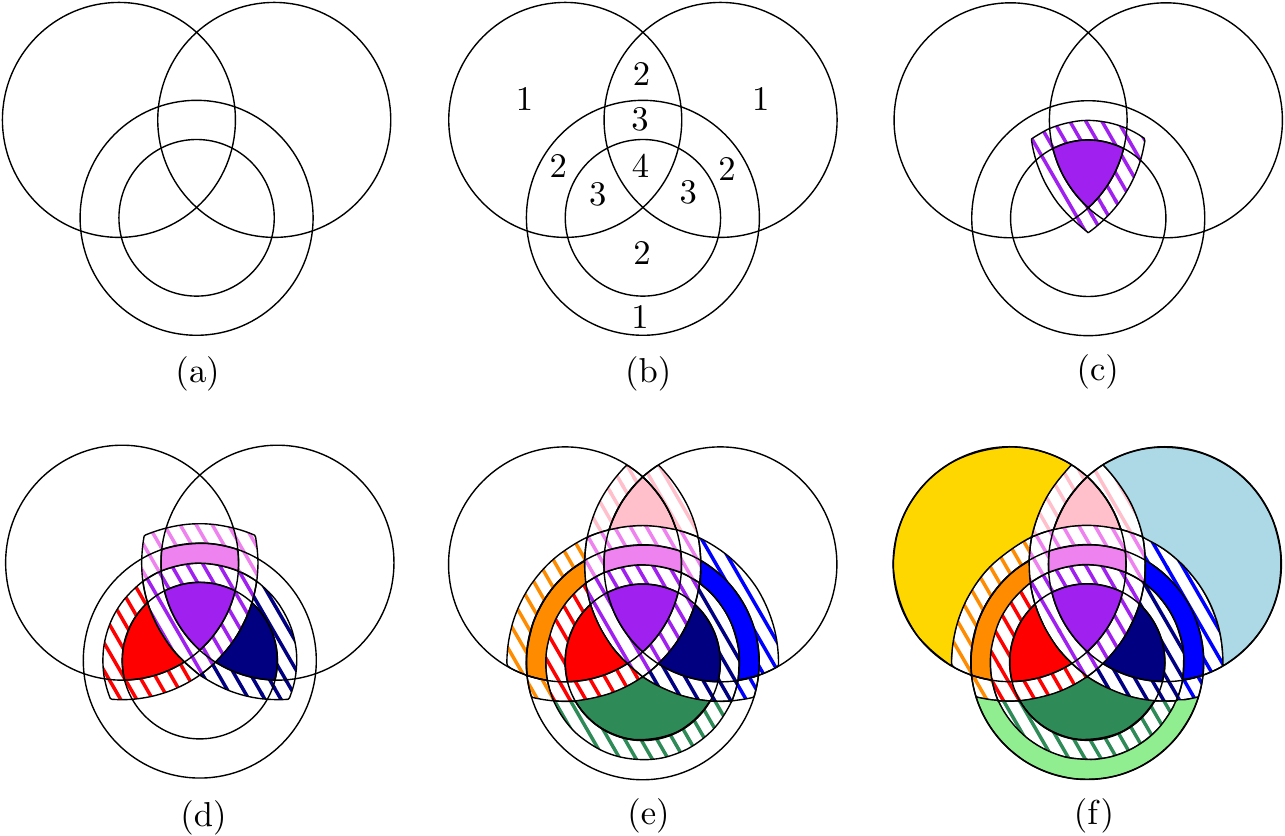}
\caption[Visual representation of general clustering algorithm]{Example of our algorithm for general cluster processing. In (a) we see the initial state of our cluster of $4$ queries. In (b) we see the calculated depths. From (c)-(f) we see in color the part we are processing, and in falling color pattern the cover we end up actually getting. This example shows that instead of doing the greedy algorithm 4 times, we end up doing it 11 times. However, the total size of the data that we are performing our algorithm on is much smaller than doing greedy 4 times because of the overlap in the queries (our algorithm never processes the same part more than once)}
\label{fig:mq_processing}
\end{figure}



\section{Query Processing in Real-time}\label{sec:real-time}\label{Ch:RT}

In our handling of real-time processing, we assume that we know everything about a certain fraction of the incoming queries beforehand (call
this the \emph{pre-real-time} set), and we get information about the remaining queries only when these queries arrive (the \emph{real-time} set). In other words, we have 
no information about future queries. To process query in real time the algorithm uses clusters formed in pre-processing. 

\subsection{The Real-time Algorithm}
Our strategy in approaching this problem is to take advantage of the real-time applicability of the $\mathtt{simpleEntropy}$ clustering algorithm. From Section~\ref{subsec:realtime_app}, we know from experiments, that we only need to process a small fraction of incoming queries to generate most of our clusters. Thus, we cluster the pre-real-time set of queries, and run one of the 
$\mathtt{GCPA}$ algorithms on the resulting clusters, storing some extra information which will be explained below. Then, we use this stored information to process the incoming 
queries quickly with a degree of optimality, as we explain below. 

We start by recalling the definition of a \emph{data part} (from Section~\ref{sec:GCPA}) and defining the related \emph{G-part}. Given a cluster $K$ and subset of queries in that cluster $P$, a \emph{data part} is the set of all elements
in 
the intersection of the queries in $P$ but not contained in any of the queries in $K\setminus P$. This implies that all the elements in the same data part have the same depth in the cluster. Figure~\ref{fig:partsnatalia} helps explain the 
concept. 

\begin{figure}[tb]
\centering
\includegraphics[width=.45\textwidth]{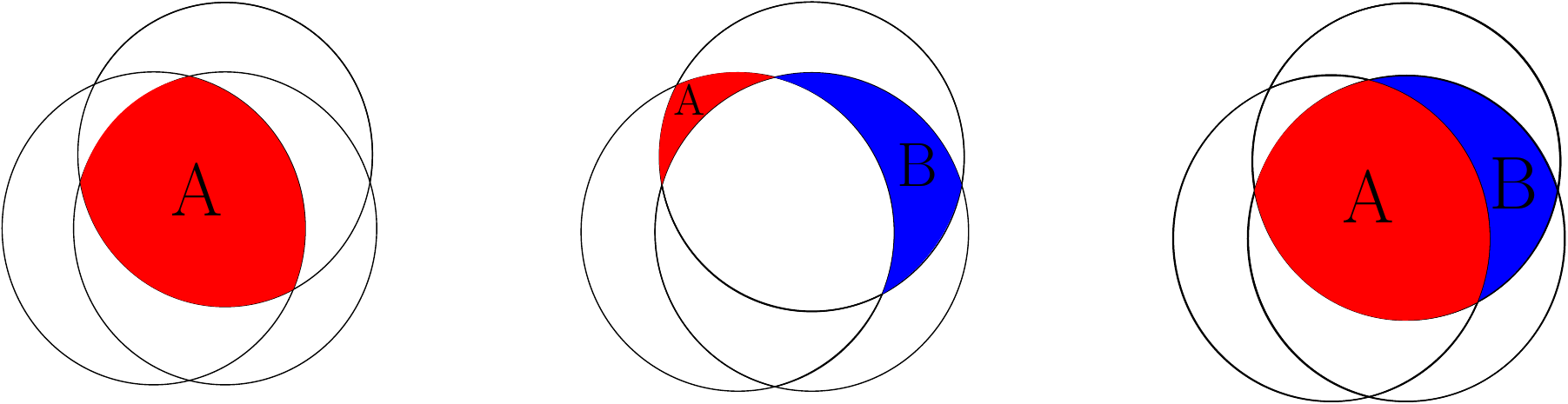}
\caption[Algorithm for query processing in real-time]{In the first picture the red area (A) is a part. In the second picture (A) and (B) areas are different parts, because, although they have the same depth, they are made from the intersection of different queries. In the third picture (A) and (B) areas are different parts, as they have different depth.}
\label{fig:partsnatalia}
\end{figure}

After running one of the $\mathtt{GCPA}$ algorithms, the cluster is processed from largest to smallest depth. The \emph{G-part} $p^{g}_{i}$ is the set of elements in the cover produced when
$\mathtt{GCPA}$ covers all elements in part $p_{i}$ that are not in any previous G-part. Note that G-parts also partition our cluster. 

To manage processing queries in real time, the algorithm makes use of queries previously covered in clustering process.  An array $T$ is created such that for each data item it stores the G-part containing this data item. In other words, element $T[i]$ is the G-part containing element $i$. For each G-part, we also store the machines that cover that G-part (this information is calculated and stored when $\mathtt{GCPA}$ is run on the non-real-time queries). The last data structure used in this algorithm is a hash table $H$, which stores, for each data element, a list of machines covering this data item. In each step the algorithm checks which G-part contains each data item and then for each G-part it checks which machines cover this G-part. Then those machines are added to the set of solutions (if they are not yet in this set). Then for each data item, which was not taken into consideration in any G-part, the algorithm checks in a hash table if any of the machines in the set of solutions covers the chosen data element. In the last step, we cover any still uncovered elements with the greedy algorithm. 
Data elements on which the greedy algorithm is run form a new G-part. 


\eat{\begin{proposition}
The real-time algorithm works in $O(N)$ time, where $N$ is the number of queries coming in real-time. 
\end{proposition}

\begin{proof}
For each query the algorithm works in constant time. The first step, while finding a G-part containing each data item, takes $O(k)$, where $k$ is the length of the query. Then the algorithm checks needed machines to cover the 
necessary G-parts, which also takes $O(k)$. In the last steps the algorithm checks if each of the uncovered elements in the query has been already covered, and then runs the greedy algorithm on the remaining data items. 
This last step also takes $O(k)$ time. In that way for each query, the algorithm has $O(k)$ time, where $k$ is bounded by a known constant (e.g. 15 in our case). For $N$ queries the real-time query processing algorithm 
thus works in $O(N)$ time. 
\end{proof}}

\eat{\subsection{Modified Clustering for real time queries}}
When a query $Q$ of a length $k$ comes in, the goal is to quickly put $Q$ into its appropriate cluster, so the above algorithm can be run. Since the greedy algorithm is linear in the length of the query, 
if it takes more than linear time to put a query into a cluster, our algorithm would be slower than just running the greedy algorithm on the query. 
\eat{\begin{proposition}
Assigning a query of size $k$ to a cluster takes $O(k^2)$ time. 
\end{proposition}
\begin{proof}
A query could have $O(k)$ potential clusters it could be put into (since we assume each element is in some 
bounded constant number of clusters, see proof of Proposition~\ref{prop:runtime}). Then, calculating the potential expected entropy of putting a query in each of the clusters also takes 
$O(k)$ time (since the probability of each data item must be calculated). Thus, with $O(k)$ calculations for each 
of $O(k)$ clusters, assigning a query to a cluster takes $O(k^2)$. 
\end{proof}}Thus, we need to develop a faster method of putting a query into a cluster. We implemented a straightforward solution. 
Instead of looking at all $O(k)$ potential clusters a query could go into, we just choose one of the elements of 
the query at random, and choose one of the potential clusters it is in at random. We call this the 
\emph{fast} clustering method, as opposed to the \emph{full} method (which is $O(k^2)$).

\section{Experimental Evaluation}
\label{sec:experiments}
\subsection{Setup}
\subsubsection{Datasets}
We run our experiments on both synthetic and real-world datasets. The sizes considered in this work are the following: each data unit is replicated $3$ times and we consider cluster of 50 homogenous machines.

\topic{Synthetic Dataset:}
The total number of data items that we consider is $100$K. We generate about $50$K queries with certain correlation between them, and each query accesses between $6$ and $15$ data items. We note that all experiments in this section are done by averaging the results from $1$M runs. Following is an explanation of the correlated query generation:

\textit{Correlated Query Workload Generation:}
\label{sec:querygen}
A sample set of queries is needed to test the effectiveness of a set cover algorithm. As mentioned in Section~\ref{sec:background} the data is distributed randomly on the machines and the queries are correlated. To generate these queries we use \emph{random graphs}. In this context, vertices counted by $n$ represent data and edges represent relations of the data.

A random graph is a graph with $n$ vertices such that each of the ${n}\choose{2}$ 
possible edges has a probability $p$ of being present. We choose $n$ and $p$ 
such that $np < 1$, since, in the Erd\H{o}s-R\'{e}nyi model \cite{erd60}, this gives a graph with many small connected components. 
The $np < 1$ case is helpful in this setting because the graph is naturally partitioned into several components as expected. 
The random graph could, for example, represent a database that contains data from many organizations, so each organization's data is a connected component and  is separate from the others. We use a modified DFS algorithm on the random graph to 
generate nearly highly correlated random queries. 

 According to the Erd\H{o}s-R\'{e}nyi Model for a random graph \cite{erd60}, depending on the value of $n$ and $p$, there are three main regimes as shown in Figure~\ref{fig:regimes}\eat{ (see \cref{tab:inters})}:
 \squishlist
 \setlength\itemsep{0em}
 \item if $np<1$ then a graph will almost surely have no connected components with more than than $O(\log{n})$ vertices,
 \item if $np=1$ then a graph will almost surely have a largest component whose number of vertices is of order $n^{\frac{2}{3}}$
 \item if $np>1$ then a graph will almost surely have a large connected component and no other component will contain more then $O(\log{n})$ vertices.
 \squishend
 \begin{figure}[ht]
 \centerline{\includegraphics[width=0.48\textwidth]{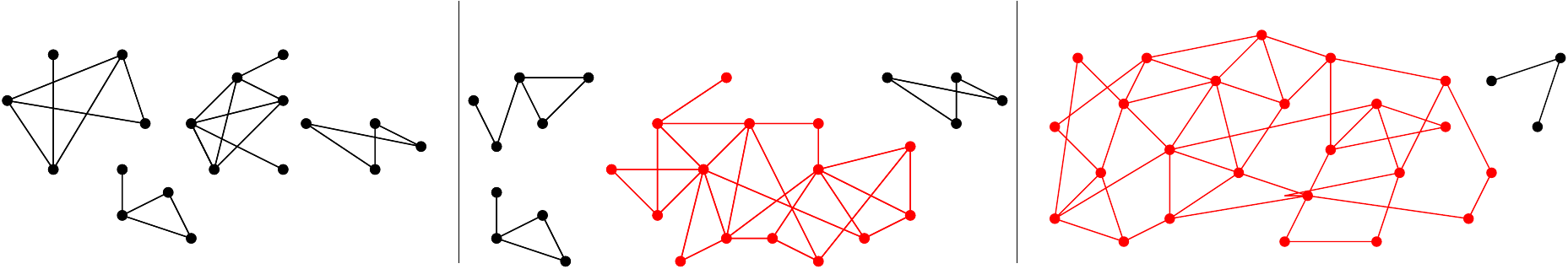}}
 \caption[The three regimes of the random graph model]{The first regime models a graph when $np<1$, middle when $np=1$ and the last one when $np>1$}
 \label{fig:regimes}
 \end{figure}
 The first regime was used because generating queries from smaller connected components theoretically makes intersections between them more probable. As it is expected that the queries have some underlying structure, the random graph model is an appropriate method of generating queries. The $np < 1$ case is helpful in this setting because the graph is naturally partitioned into several components as expected. The random graph could, for example, represent a database that contains data from many organizations, so each organization's data is a connected component and  is separate from the others. 
\begin{algorithm}
\caption{Query Generation Algorithm}\label{euclid}
\begin{algorithmic}[1]
\Require $N_Q$, $Q=\emptyset$
\For{$|Q|\neq N_Q$}
\State $q=\emptyset$, $K=\emptyset$
\State $l \gets rand(minQueryLen, maxQueryLen)$
\State $x = rand(v \in V)$
\State $q \gets q \cup x$
\State $K \gets K \cup \{v~|~(v, x) \in E, v\in V, x\in V\}$
\While{$|q| \neq l$}
\State $x = rand(v \in K)$
\If{$x\notin q$}
\State $q \gets q \cup x$
\EndIf
\EndWhile
\State $Q\gets Q\cup q$
\EndFor
\end{algorithmic}
\label{alg:QG}
\end{algorithm}
\texttt{QueryGeneration} algorithm (Algorithm~\ref{alg:QG}) is as follows. First, build a random graph $G$ with $np<1$. \eat{Query generation technique is shown in Algorithm~\ref{alg:alg1}. }Key idea here is to generate random subgraph with number of nodes equal to $l$ where $l$ is the query length, such that $(minQueryLen\leq l \leq maxQueryLen)$. Repeat this till we generate desired number of queries. \eat{In this work, we generate about 35K queries. Each query touches about $6$ to $15$ number of data items. \texttt{Query Generation} algorithm is used to generate nearly 33000 queries and}To assess the quality of our synthetic query workload generator, every query is compared with every other query to determine the size of intersections. Then the same number of queries were generated  uniformly randomly and pairwise intersections were calculated. As expected, queries generated using \texttt{QueryGeneration} algorithm have much more intersections then queries generated randomly. \eat{See \cref{tab:inters} for the results.}

\topic{Real-world Dataset:} We consider TREC Category B Section 1 dataset which consists of 50 million English pages. For queries we consider 40 million AOL search queries. In order to process these 50 million documents to document shards, we perform K-means clustering using Apache Mahout where K=10000. We consider each document shard as a data item in this paper. These document shards are distributed across 50 homogenous machines and are 3-way replicated. Each AOL web query is run through Apache Lucene to get top 20 document shards. Then we run our incremental set cover based routing to route queries to appropriate machines containing relevant document shards.\eat{Also, we divide our query dataset into two parts, first part has 30 million queries that is used as history for analyzing and applying techniques discussed in this paper. Then we evaluate our techniques on test dataset which is the second part containing 10 million queries.}

Overall, we evaluate our algorithms on a set of $50$K synthetically generated queries generated from a graph with $np = .993$ and on real-world dataset. $20$K queries from synthetic dataset and $8$M queries from real-world dataset among them are used to create clusters and our routing approach is tested on remaining $30$K queries from synthetic dataset and $32$M queries from real-world dataset.

\subsubsection{Baseline} 
\label{sec:base}
When a query $Q$ is received a request is sent to all machines that contains an element of $Q$. The machines are added to the set cover by the order in which they respond, until the query is covered. The first machine to respond is automatically added to the cover. The next machine to respond is added if it contains any element from the query that is not yet in the cover. This process is continued until all elements of the query $Q$ are covered. We call this method \textit{baseline} set covering. While the method is fast, there is no discrimination in the machines taken, which means 
that the solution returned is far from optimal, as the next example illustrates: 

Consider a query, $Q = \{1, \dots n\}$, and a set of machines, $M_1, \dots, M_{n+1}$, where $ M_i = \{i\}$ for $ i\leq n$ and $M_{n+1} = \{1, \dots n\}$. If machines $M_1, \dots , M_n$ respond first then this algorithm will cover $Q$ with $n$ machines where the optimal cover contains only one machine, $M_{n+1}$. 
Given N queries, our algorithm should improve upon the average optimality offered by this baseline covering method and be faster than running the greedy set cover $N$ times.

\subsubsection{Machine} 
The experiments were run on a Intel Core i7 quad core with hyperthreading CPU 2.93GHz, 16GB RAM, Linux Mint 13. We create multiple logical partitions within this setup and treat each logical partition as a separate machine. This particular experimental setup does not undermine the generality of our approach in anyway. Our results and observations stand valid even when run on distributed setup with real machines.

\begin{figure*}
\centering
\subfigure[Run-time for synthetic dataset]{ \includegraphics[width=.22\textwidth]{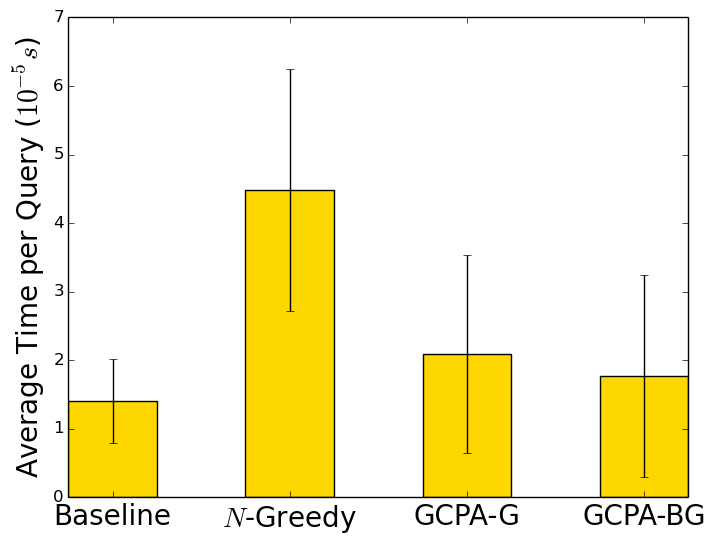} \label{fig:times}}
\hfill
\subfigure[Run-time for real-world dataset]{ \includegraphics[width=.22\textwidth]{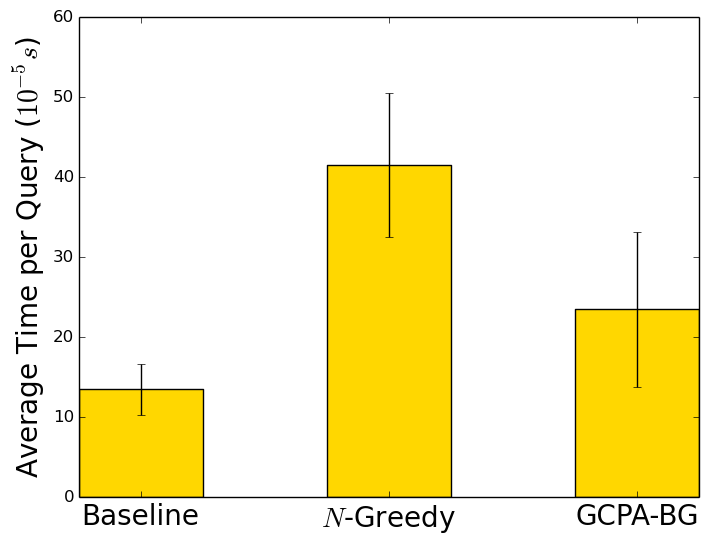}\label{fig:times-real}}
\hfill
\subfigure[Optimality for synthetic dataset]{ \includegraphics[width=.22\textwidth]{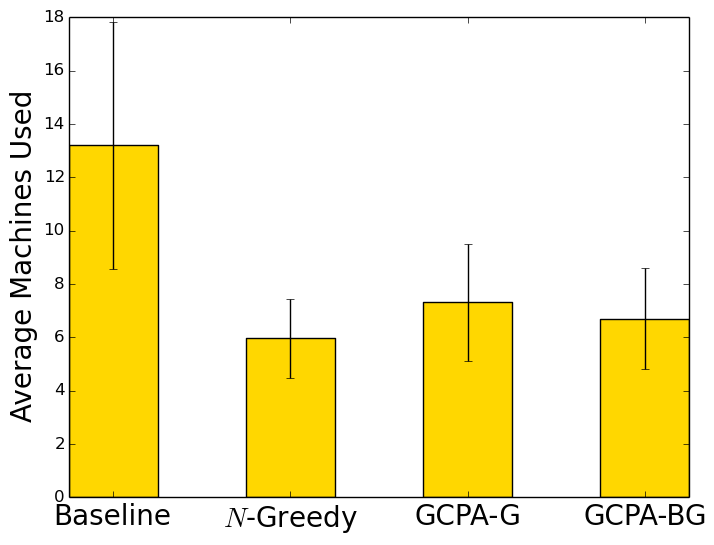}\label{fig:optimality}}
\hfill
\subfigure[Optimality for real-world dataset]{ \includegraphics[width=.22\textwidth]{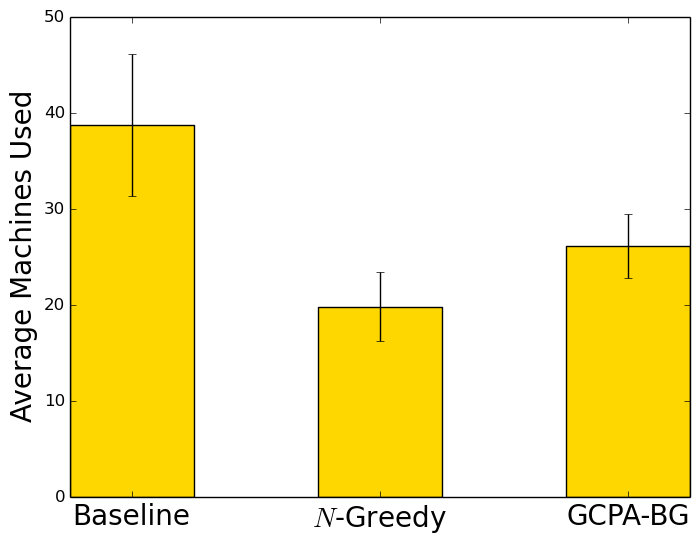}\label{fig:optimality-real}}
\caption[Comparison of run-times and optimality]{Comparison of run-time and optimality (average {\em query span}) of our algorithms on synthetic dataset and real-world dataset. 
\eat{and the three reference algorithms on a set of $50,000$ queries generated 
from a graph with $np = .993$. Our algorithms perform considerably
faster than $\mathtt{N\_Greedy}$ and are also faster than the 
smarter baseline algorithm. In terms of optimality, both our algorithms considerably 
outperform the standard baseline algorithm. Error bars for optimality represent 1 standard deviation.
Error bars for the timing graph are omitted since queries are not processed 1-by-1 with 
$\mathtt{GCPA}$.}}
\label{fig:times_and_optimality}
\end{figure*}

\subsection{Experimental Analysis of Clusterings}
We ran the the clustering algorithm on several sets of queries. All 
of these query sets are of size $50,000$ and are generated via the 
Erd\H{o}s-R\'{e}nyi graph regime with $0.9 < np < 1.0$ according to 
the query generation algorithm described in Section~\ref{sec:querygen}. 
We specifically tested the resulting clusters for quality of clustering
and for its applicability towards real-time processing. The ideas behind
real-time processing are more thoroughly discussed in 
Section~\ref{Ch:RT}, but the essential idea is this: given a small 
sample of queries beforehand for pre-computing, we need to be able to 
process new queries as they arrive, with no prior information about them. 

\subsubsection{Clustering Quality}
In a high-quality cluster most of the data elements 
have probability close to $1.0$. Intuitively, this indicates that the 
queries in the cluster are all extremely similar to each other (i.e. 
they all contain nearly the same data elements). Then one measure of 
clustering quality would be to look at the probability of data elements
across clusters. 

As a first measure, we recorded the probability of each data element in 
each cluster, and Figure~\ref{fig:cluster_prob_hist} depicts the results in 
a histogram for a typical clustering. The high frequency of data 
elements with probability over $0.9$ indicates that a significant number
of data elements have high probability within their cluster. (Note that 
in this analysis, if a data element is contained in many clusters, its
probability is counted separately for each cluster.) However, 
interestingly, the distribution then becomes relatively uniform for 
all the other ranges of probabilities. This potentially illustrates the difference between mediocre and 
high-quality clusters described in Section~\ref{top:cluster}. 
A more ideal clustering algorithm
would increase the number of data elements in the higher probability bins and
decrease the number of those in lower quality bins. Still, the 
prevalence of elements with probability greater than $0.9$ is heartening
because this indicates a fairly large common intersection among all the
clusters. By processing this common intersection alone, we 
potentially cover a significant fraction of each query in the cluster 
with just a single greedy algorithm.

To paint a broad picture, the above measure of cluster quality ignores the clusters themselves. There may be variables inherent to the cluster
which affect its quality and are overlooked. For example, perhaps 
clusters begin to deteriorate once they reach a certain size. 

Let us define the \emph{average probability} of cluster $K$ as: 
\begin{equation}
\overline{p}(K) = \frac{1}{\sum_{Q\in K}|Q|}\sum_{Q\in K}\sum_{x\in Q}p_x
\end{equation}
Essentially, $\overline{p}(K)$ is a weighted average of the 
probabilities of each element in the cluster (so data elements that are 
in many queries are weighted heavily). In Figure~\ref{fig:bycluster_probs}
we see that there is some deterioration of average probability as 
the clusters get larger, but for most of the size range, the quality
is well scattered. While most of the clusters have average probability 
greater than $0.6$, a stronger clustering algorithm would collapse 
this distribution upwards. 

\begin{figure}[tb]
\centering
\subfigure[Overall averages]{ \includegraphics[width=.22\textwidth]{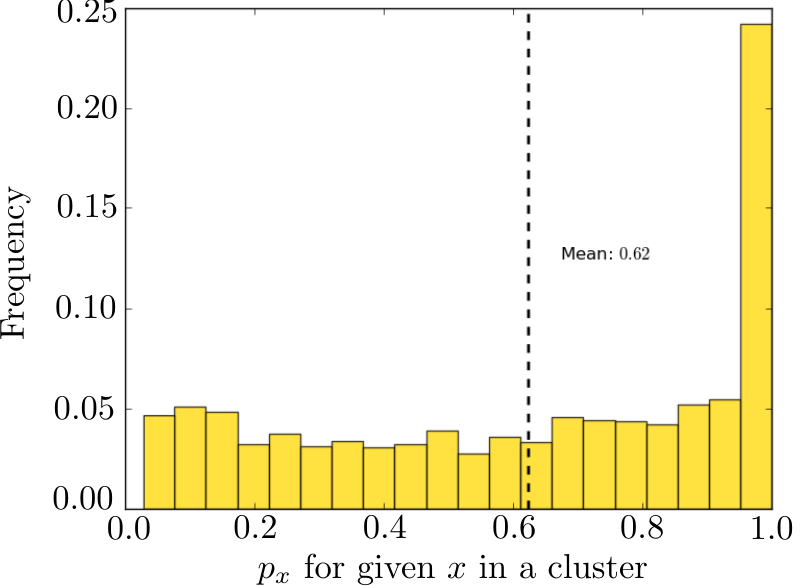} \label{fig:cluster_prob_hist}}
\hfill 
\subfigure[Weighted by cluster]{ \includegraphics[width=.22\textwidth]{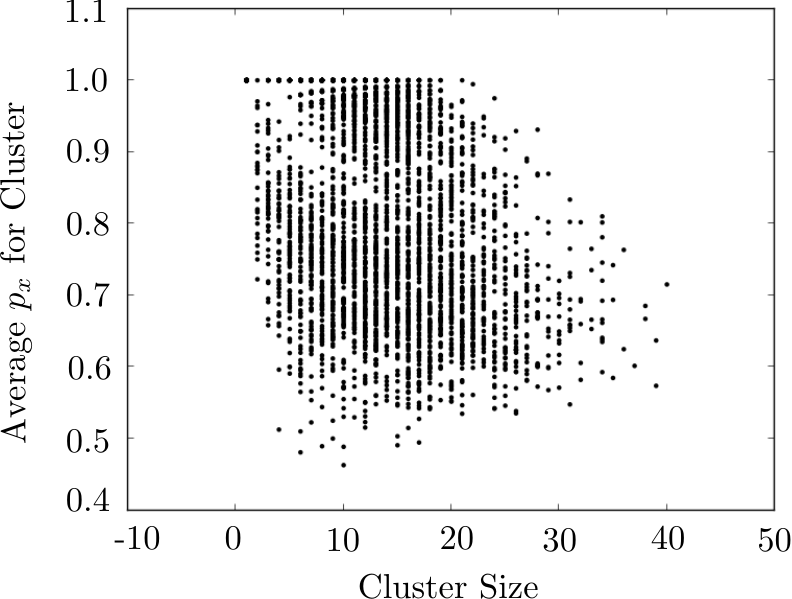} \label{fig:bycluster_probs}}

\caption[Clustering quality]{These graphs are generated by clustering $50$K synthetically generated queries for which 
$np = 0.973$. In Figure~\ref{fig:cluster_prob_hist}, we look at 
each data element in each cluster and record its probability in that cluster. A high-quality algorithm would 
have tall bars on the right and very short bars elsewhere. 
In Figure~\ref{fig:bycluster_probs}, we take a weighted average of the 
probability of each element. Note the downward trend as 
clusters get large. This may mean we should restrict cluster size.}
\label{fig:cluster_quality}
\end{figure}

\subsubsection{Real-time Applicability}
\label{subsec:realtime_app}
For $\mathtt{simpleEntropy}$ to be successful in dealing with real-time queries we
have additional requirements. First of all, the incoming queries need 
to be processed quickly. If, for example, the query needed to be checked
against every single cluster before it was put in one, simply running 
the greedy algorithm on it would probably be faster, since there are 
on the order of thousands of clusters. \eat{As described in \cref{prop:runtime}, o}
The \emph{fast} version of the algorithm, which only samples one element from 
each of the real-time queries, meets this requirement. 
Second, we want most of the
cluster to be generated when only a small fraction of the queries are 
already processed. This way, most of the information about incoming 
queries is already computed, which allows us to improve running time. 
In Figure~\ref{fig:cluster_progress} and Table~\ref{tab:cluster_progress}, we see that more than 75\% of the total clusters 
are generated with only 20\% of the data processed, which means our clustering 
algorithm is working as we want it to. Finally, we want incoming queries
to contain most of the high probability elements of the cluster. 
Specifically, let $Q_1, \dots, Q_n$ be the queries in a cluster $K$, 
where $X = Q_1 \cap \dots \cap Q_n$, and let $Q^*$ be the incoming 
query. We want $X\subset Q^*$, since this means that the deepest G-part can cover a
lot of $Q$ with little waste. While this seems to be true for the high-quality 
clusters described above, we could seek to improve our algorithm to generate fewer mediocre clusters. 

\begin{figure}[tb]
\centering
\includegraphics[width=.35\textwidth]{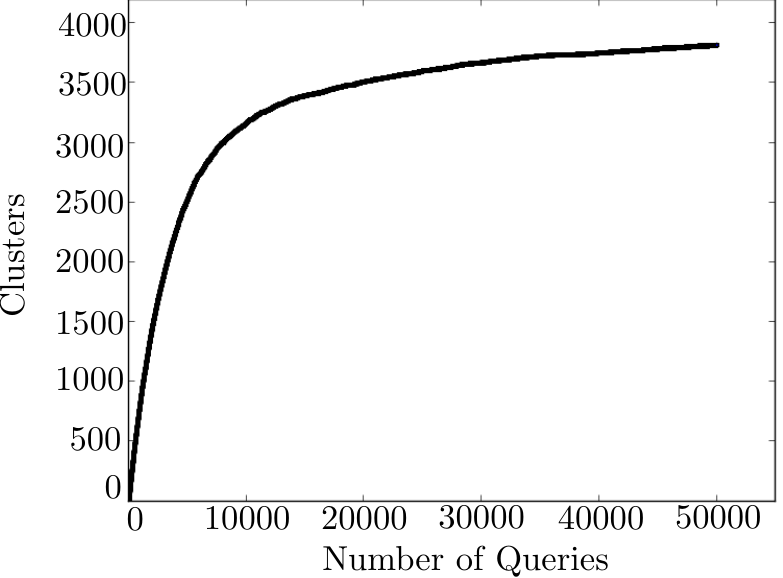}
\caption[Cluster generation as queries come in]{We plot the number of clusters as each of $50$k queries comes in, after being generated from an Erd\H{o}s-R\'{e}nyi graph with 
$np = .999$.}
\label{fig:cluster_progress}
\end{figure}

\eat{\begin{table}[tb]
\tiny
\centering

\begin{tabular}{ccccccccccccc}
$\%$ Queries Processed & $\%$ Clusters Formed \\
\hline
 6.0 & 50.0 \\
10.0 & 66.1 \\
13.8 & 75.0 \\
25.0 & 86.6 \\
33.7 & 90.0 \\
40.0 & 91.9 \\
50.0 & 94.3 \\
53.7 & 95.0 \\
75.0 & 97.9 \\
88.2 & 99.0 \\
90.0 & 99.2 \\
99.5 & 99.9 \\
\hline
\end{tabular}
\vspace{3pt}
\caption[Clustering thresholds for real-time]{These results can help select the threshold for pre-processing queries to 
form most clusters. The data are taken from processing 50,000 queries from a random graph 
with $np = .999$. Potential thresholds include $13.8\%, 33\%, 40\%$.}
\label{tab:cluster_progress}
\end{table}}

 \begin{table}[tb]
 \footnotesize
 \centering
 \setlength\tabcolsep{1.2pt}
 \begin{tabular}{c||c|c|c|c|c|c|c|c|c|c|c|c|}
 \hline
 $\%$ Queries Processed & 6.0 & 10.0 & 13.8 & 25.0 & 33.7 & 40.0 & 50.0 & 53.7 & 75.0 & 88.2 & 90.0 & 99.5 \\
 \hline
 $\%$ Clusters Formed & 50.0 & 66.1 & 75.0 & 86.6 & 90.0 & 91.9 & 94.3 & 95.0 & 97.9 & 99.0 & 99.2 & 99.9 \\
 \hline
 \end{tabular}
 \vspace{3pt}
 \caption[Clustering thresholds for real-time]{These results can help select the threshold we need for pre-processing queries to 
 form most of our clusters. The data are taken from processing 50,000 queries from a random graph 
 with $np = .999$. Potential thresholds include $13.8\%, 33\%, 40\%$.}
 \label{tab:cluster_progress}
 \end{table}

 \subsection{Experimental Comparison of Cluster Processing Algorithms}
\label{sec:experimental_nrt}
For comparing our cluster processing algorithms we had implemented two reference algorithms and two that we developed ourselves. We show in this section that our algorithms are successful, in that they are both fast and optimal. 

\eat{The first reference algorithm is the baseline algorithm, mentioned in \cref{sec:base}. While implementing that algorithm, we have realized that a much better performing algorithm could be used, while retaining the same time performance. We call it $\mathtt{BetterBaseline}$. At each step, $\mathtt{BetterBaseline}$ chooses a random uncovered data unit and then a random machine that covers that data unit, and then checks what that machine covers and updates the uncovered part. This way, instead of possibly having chosen all the machines that have a nonempty intersection with our query, which might happen using the baseline algorithm, we can have at most $|Q|$ chosen machines. This algorithm is far ``smarter'' than the methods normally used as the baseline, since the algorithm only looks at machines that actually intersect the query. Usually, a query is 
sent to all machines, and whichever respond first are used as the cover for the query\eat{ \cite{kul15}}. While the baseline algorithms are fast, the covers they return are far from optimal. Thus, the goal is for our algorithms to 
generate more optimal covers than the baseline algorithms. }

The first reference algorithm is the one primarily evaluated in the papers by Kumar and Quamar et al.,~\cite{Kumar:2014:SWD:2691523.2691545}\cite{Quamar:2013:SSW:2452376.2452427}\cite{ashwin_phd}, we call it $\mathtt{N\_Greedy}$. This is simply running the greedy algorithm on each query independently. This algorithm has the opposite properties of the baseline algorithms. While its covers are as close 
to optimal as possible, it has a longer run-time than the baseline. Thus, we want our algorithm to run faster than $\mathtt{N\_Greedy}$. 

The two algorithms that we have developed and implemented are the $\mathtt{GCPA}$ with the greedy algorithm (\eat{we call this }$\mathtt{GCPA\_G}$) and $\mathtt{GCPA}$ with $\mathtt{BetterGreedy}$ (\eat{we call this }$\mathtt{GCPA\_BG}$). The major difference between the reference algorithms and our algorithms is that we are using clustering to exploit the correlations
and similarities of the incoming queries. Our algorithms are faster than $\mathtt{N\_Greedy}$ and more optimal than the baseline algorithms.

We compare the run-time and optimality (average number of machines that a query touches) of our algorithms
and the two reference algorithms on \eat{a set of $50$K synthetically generated queries generated 
from a graph with $np = .993$ and on real-world dataset. $20$K queries from synthetic dataset and $8$M queries from real-world dataset among them are used to create clusters and our routing approach is tested on remaining $30$K queries from synthetic dataset and $32$M queries from}both synthetic and real-world datasets. Our algorithms perform considerably
faster than $\mathtt{N\_Greedy}$ and are also faster than the 
smarter baseline algorithm. In terms of optimality, both our algorithms considerably 
outperform the standard baseline algorithm as shown in Figure~\ref{fig:times_and_optimality}\eat{, the results are favorable}. In summary, when evaluating with synthetic dataset, as shown in Figures~\ref{fig:times} and \ref{fig:optimality}, our technique is about $2.5\times$ faster when compared to repeated greedy technique $\mathtt{N\_Greedy}$ and selects $50\%$ fewer machines when compared to baseline routing technique. On the other hand, we evaluate on GCPA-BG for the real-world dataset because it has better optimality, and in the real-time case, the time penalty for using GCPA-BG over GCPA-G is only relevant in the pre-computing stage. For real-world dataset case, as shown in Figures~\ref{fig:times-real} and \ref{fig:optimality-real}, our technique is about $2\times$ faster when compared to repeated greedy technique $\mathtt{N\_Greedy}$ and selects $32\%$ fewer machines when compared to baseline routing technique. \eat{Error bars for performance and optimality represent 1 standard deviation}The error bars shown are one standard deviation long.
\eat{Error bars for the timing graph are omitted since queries are not processed 1-by-1 with 
$\mathtt{GCPA}$.}\eat{We have tested all the algorithms in terms of processing time per query and average load (number of machines accessed) per query.} Even though 
the figures show the results for only one set of queries, we have run dozens of samples, and the overall picture is the same. 
The results of our experiments provide strong indication that our algorithm is indeed
an effective method for incremental set cover, in that it is faster than $\mathtt{N\_Greedy}$
and more optimal than the baseline. 

In terms of optimality, it is also important to do a pairwise comparison of cover lengths (i.e. does our algorithm perform better for queries of any size). 
Taking the average as we have done in Figure~\ref{fig:optimality} masks potentially important variation. 
We want to ensure that our algorithm effectively handles queries of all sizes. In Figures~\ref{fig:bg_ng} and \ref{fig:lg_ng}, we compare the query-by-query performance (in terms of optimality) 
of our two algorithms against $\mathtt{N\_Greedy}$ for the synthetic dataset.  
The $x$-axis is the number of sets required to cover a 
query using $\mathtt{N\_Greedy}$. The $y$-axis, ``$\Delta$ Cover Length",
is the length of the cover given by our algorithm minus the length of the 
greedy cover. 
The number next to the $y$-axis at $y=k$ shows the normalized proportion 
of queries for which 
the $\mathtt{GCPA}$ cover is at most $k$ machines larger than the 
$\mathtt{N\_Greedy}$. The size of the circle indicates the number of 
queries at that coordinate. 
With $\mathtt{GCPA\_BG}$, we see that more than $90\%$ of all queries are covered with at most one more machine than the greedy 
cover, and for the majority of queries, the covers are the same size. The $\mathtt{GCPA\_G}$ algorithm does not perform quite as well. Even in this case, the majority of queries 
are covered using only one more machine than the greedy cover. Since $\mathtt{GCPA\_BG}$ is slower than $\mathtt{GCPA\_G}$, users can choose their algorithm based on their preference for speed or optimality. 

On the other hand, in Figure~\ref{fig:lg_ng_real}, we evaluate the performance of our real-time algorithm on the real-world dataset on a query-by-query basis. For each query, we record the number of machines used to cover it using our algorithm and the number of machines required to cover it using the baseline algorithm and record the difference, i.e. the ``$\Delta$ Cover Length'' on the $y$-axis is the size of the baseline cover minus the size of our algorithm's cover. The area of the circle at point $(x,y)$ is proportional the number of queries for which our algorithm used $x$ machines to cover and for which the difference in cover length is $y$. Thus, the total area of the points above the $y=0$ line represents where our algorithm outperforms the baseline algorithm. We see in the figure that the vast majority (96.5\%) of the queries are covered more efficiently by our algorithm than by the baseline algorithm. This is actually a bit of an understatement because most of the queries for which the baseline algorithm performed with better optimality (i.e. points below $y=0$) were queries of length one (or at least very small), which can easily be handled as special cases. In this case, only one machine is required to cover, but our algorithm takes the entire cover from the cluster that the length one query was put into. This situation is easily remedied. For small enough queries, especially queries of length one, instead of running our algorithm we should just cover them directly.
\begin{figure*}
\centering
\subfigure[$\mathtt{GCPA\_DL}$ vs. $\mathtt{N\_Greedy}$ for synthetic dataset]{ \includegraphics[width=.3\textwidth]{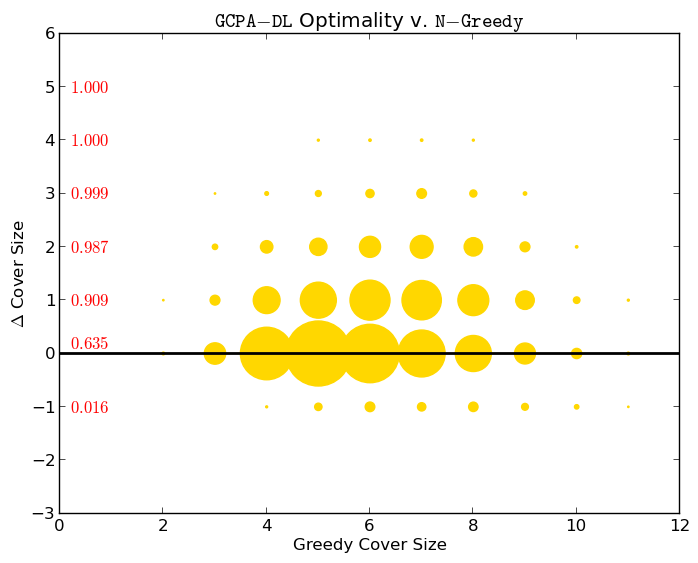} \label{fig:bg_ng}}
\hfill
\subfigure[$\mathtt{GCPA\_G}$ vs. $\mathtt{N\_Greedy}$ for synthetic dataset]{ \includegraphics[width=.3\textwidth]{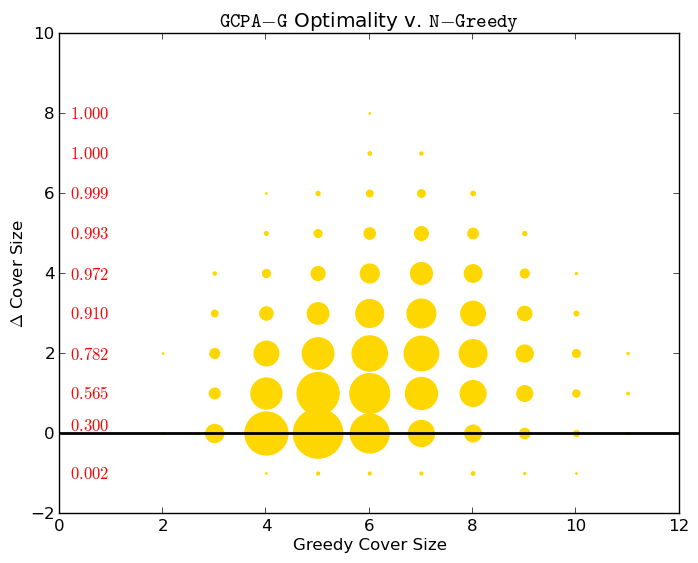} \label{fig:lg_ng}}
\hfill
\subfigure[$\mathtt{GCPA}$ vs. $\mathtt{Baseline}$ for real-world dataset]{ \includegraphics[width=.3\textwidth]{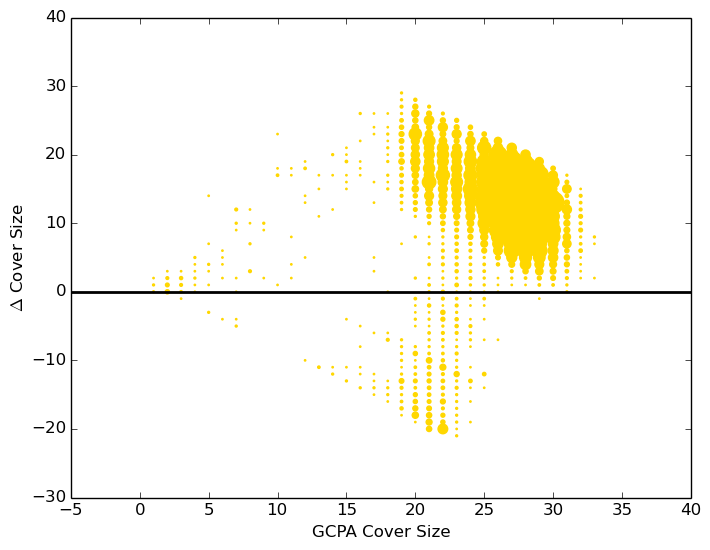} \label{fig:lg_ng_real}}
\caption[Pairwise optimality comparisons]{ 
Pairwise comparisons of optimality for our algorithms. 
}
\end{figure*}
\eat{It is important here that we describe more precisely what we mean when we time our algorithms. 
Our algorithms require a significant amount of pre-processing (the clustering and then breaking up the cluster into parts as described above) which we do \emph{not} take into account during the timing. We only 
count the segments of the algorithm that actually access machines when we are timing. 
The remark is not an ad-hoc decision. The pre-computation can be done well before any queries actually need to be processed, and can thus be amortized over an arbitrary amount of time. It is the accessing of the machines that is 
costly. In the real-time case discussed in the following chapter, we have to reconsider our timing methods since there is no time for pre-computation with a query that comes in real-time.

As required,}In conclusion, we have delivered an algorithm that is significantly faster than $\mathtt{N\_Greedy}$ and also more optimal than the baseline algorithm. 


\section{Conclusion}
In this paper, we presented an efficient routing technique using the concept of {\em incremental set cover} computation. The key idea is to reuse the parts of {\em set cover} computations for previously processed queries to efficiently route real-time queries such that each query possibly touches a minimum number of machines for its execution. To enable the sharing of {\em set cover} computations across the queries, we take advantage of correlations between the queries by clustering the known queries and keeping track of computed query {\em set covers}. We then reuse the parts of already computed {\em set covers} to cover the remaining queries as they arrive in real-time. We evaluate our techniques using both real-world TREC with AOL datasets, and simulated workloads. Our experiments demonstrate that our approach can speedup the routing of queries significantly when compared to repeated greedy {\em set cover} approach without giving up on optimality. We believe that our work is extremely generic and can benefit variety of scale-out data architectures such as distributed databases, distributed IR, map-reduce, and routing of VMs on scale-out clusters.



\bibliographystyle{IEEEtran}
\scriptsize
\bibliography{AA-Bibliography/Biblio}

\begin{thebibliography}{10}
\providecommand{\url}[1]{#1}
\csname url@samestyle\endcsname
\providecommand{\newblock}{\relax}
\providecommand{\bibinfo}[2]{#2}
\providecommand{\BIBentrySTDinterwordspacing}{\spaceskip=0pt\relax}
\providecommand{\BIBentryALTinterwordstretchfactor}{4}
\providecommand{\BIBentryALTinterwordspacing}{\spaceskip=\fontdimen2\font plus
\BIBentryALTinterwordstretchfactor\fontdimen3\font minus
  \fontdimen4\font\relax}
\providecommand{\BIBforeignlanguage}[2]{{%
\expandafter\ifx\csname l@#1\endcsname\relax
\typeout{** WARNING: IEEEtran.bst: No hyphenation pattern has been}%
\typeout{** loaded for the language `#1'. Using the pattern for}%
\typeout{** the default language instead.}%
\else
\language=\csname l@#1\endcsname
\fi
#2}}
\providecommand{\BIBdecl}{\relax}
\BIBdecl

\bibitem{Kumar:2014:SWD:2691523.2691545}
K.~A. Kumar, A.~Quamar, A.~Deshpande, and S.~Khuller, ``Sword: Workload-aware
  data placement and replica selection for cloud data management systems,''
  \emph{The VLDB Journal}, vol.~23, no.~6, pp. 845--870, Dec. 2014.

\bibitem{Quamar:2013:SSW:2452376.2452427}
A.~Quamar, K.~A. Kumar, and A.~Deshpande, ``Sword: Scalable workload-aware data
  placement for transactional workloads,'' in \emph{Proceedings of the 16th
  International Conference on Extending Database Technology}, ser. EDBT '13,
  2013, pp. 430--441.

\bibitem{Curino:2010:SWA:1920841.1920853}
C.~Curino, E.~Jones, Y.~Zhang, and S.~Madden, ``Schism: A workload-driven
  approach to database replication and partitioning,'' \emph{VLDB}, vol.~3, no.
  1-2, pp. 48--57, Sep. 2010.

\bibitem{Kulkarni:2015:SSE:2766484.2738035}
A.~Kulkarni and J.~Callan, ``Selective search: Efficient and effective search
  of large textual collections,'' \emph{ACM Trans. Inf. Syst.}, vol.~33, no.~4,
  pp. 17:1--17:33, 2015.

\bibitem{ashwin_phd}
A.~K. Kayyoor, ``Minimization of resource consumption through workload
  consolidation in large-scale distributed data platforms,'' \emph{Digital
  Repository at the University of Maryland}, 2014.

\bibitem{vaz13}
V.~V. Vazirani, \emph{Approximation Algorithms}.\hskip 1em plus 0.5em minus
  0.4em\relax Springer Science \& Business Media, 2013.

\bibitem{DBLP:journals/siamcomp/AlonAABN09}
N.~Alon, B.~Awerbuch, Y.~Azar, N.~Buchbinder, and J.~Naor, ``The online set
  cover problem,'' \emph{{SIAM} J. Comput.}, vol.~39, no.~2, pp. 361--370,
  2009.

\bibitem{unweighted}
A.~Levin, ``Approximating the unweighted k-set cover problem: Greedy meets
  local search,'' in \emph{Approximation and Online Algorithms}, 2007, vol.
  4368, pp. 290--301.

\bibitem{agg13}
C.~C. Aggarwal and C.~K. Reddy, \emph{Data Clustering: Algorithms and
  Applications}.\hskip 1em plus 0.5em minus 0.4em\relax CRC Press, 2013.

\bibitem{kle05}
J.~Kleinberg and E.~Tardos, \emph{Algorithm Design}.\hskip 1em plus 0.5em minus
  0.4em\relax Addison-Wesley Longman Publishing Co., Inc., 2005.

\bibitem{qruqlinse}
R.~Baeza-Yates, C.~Hurtado, and M.~Mendoza, ``Query recommendation using query
  logs in search engines,'' in \emph{Current Trends in Database Technology -
  EDBT 2004 Workshops}, 2005, vol. 3268, pp. 588--596.

\bibitem{1183888}
S.-L. Chuang and L.-F. Chien, ``Towards automatic generation of query taxonomy:
  a hierarchical query clustering approach,'' in \emph{Data Mining, 2002. ICDM
  2003. Proceedings. 2002 IEEE International Conference on}, 2002, pp. 75--82.

\bibitem{2002:QCU:503104.503108}
``Query clustering using user logs,'' \emph{ACM Trans. Inf. Syst.}, vol.~20,
  no.~1, pp. 59--81, 2002.

\bibitem{pas97}
V.~T. Paschos, ``A survey of approximately optimal solutions to some covering
  and packing problems,'' \emph{ACM Computing Surveys (CSUR)}, vol.~29, no.~2,
  pp. 171--209, 1997.

\bibitem{DBLP:conf/ijcai/ZhaoSXHZ15}
Z.~Zhao, R.~Song, X.~Xie, X.~He, and Y.~Zhuang, ``Mobile query recommendation
  via tensor function learning,'' in \emph{Proceedings of the Twenty-Fourth
  International Joint Conference on Artificial Intelligence, {IJCAI} 2015,
  Buenos Aires, Argentina, July 25-31, 2015}, 2015, pp. 4084--4090.

\bibitem{DBLP:conf/sigmod/GuptaKRBGK11}
N.~Gupta, L.~Kot, S.~Roy, G.~Bender, J.~Gehrke, and C.~Koch, ``Entangled
  queries: enabling declarative data-driven coordination,'' in
  \emph{Proceedings of the {ACM} {SIGMOD} International Conference on
  Management of Data, {SIGMOD} 2011, Athens, Greece, June 12-16, 2011}, 2011,
  pp. 673--684.

\bibitem{bar02}
D.~Barbar{\'a}, Y.~Li, and J.~Couto, ``{COOLCAT}: An entropy-based algorithm
  for categorical clustering,'' in \emph{Proceedings of the eleventh
  international conference on information and knowledge management}.\hskip 1em
  plus 0.5em minus 0.4em\relax ACM, 2002, pp. 582--589.

\bibitem{erd60}
P.~Erd\H{o}s and A.~R{\'e}nyi, ``On the evolution of random graphs,''
  \emph{Publications of the Mathematical Institute of the Hungarian Academy of
  Sciences}, vol.~5, pp. 17--61, 1960.

\end{thebibliography}

\end{document}